\documentclass[11pt]{article}

\usepackage{amsmath,amssymb,amsfonts,amsthm,mathtools}
\usepackage[margin=1in]{geometry}
\usepackage{graphicx}
\usepackage{booktabs}
\usepackage{natbib}
\usepackage{hyperref}
\usepackage{enumitem}
\usepackage{array}
\usepackage{authblk}
\usepackage{xcolor}
\definecolor{royalblue}{rgb}{0.25, 0.41, 0.88}
\usepackage{hyperref}
\hypersetup{
    colorlinks= true,
    linkcolor= black,
    filecolor= cyan,      
    urlcolor= black, 
    citecolor= royalblue
}
\usepackage{tikz, subcaption}
\usetikzlibrary{shapes, arrows, positioning}

\newtheorem{assumption}{Assumption}
\newtheorem{definition}{Definition}
\newtheorem{proposition}{Proposition}
\newtheorem{theorem}{Theorem}
\newtheorem{remark}{Remark}
\newtheorem{lemma}{Lemma}

\newcommand{\E}{\mathbb{E}}
\newcommand{\Prob}{\mathbb{P}}
\newcommand{\Var}{\operatorname{var}}
\newcommand{\indep}{\perp\!\!\!\perp}
\newcommand{\one}{\mathbf{1}}
\newcommand{\TV}{\operatorname{TV}}
\newcommand{\dd}{\,\mathrm{d}}

\title{\bf Model-Robust Direct Effect Under Confounder--Mediator Ambiguity\vspace{5mm}}

\author[]{AmirEmad Ghassami\thanks{Email: \texttt{ghassami@bu.edu}}}

\affil[]{Department of Mathematics and Statistics, Boston University}
\date{}

\begin{document}

\maketitle

\begin{abstract}
Direct effect analyses usually require deciding whether a focal variable is a pre-exposure confounder or a post-exposure mediator. In observational studies, that distinction may be unclear because timing is measured coarsely or the variable reflects an evolving process. Considering the average treatment effect (ATE) and the natural direct effect (NDE) as a common notion of the direct effect when the focal variable is a confounder and a mediator, respectively, we show that, in general, no single observed-data estimand recovers both the ATE when the focal variable is a confounder and the NDE when it is a mediator. Consequently, if a practitioner applies an NDE estimator when the variable is actually pre-exposure, the resulting estimate may have no clear causal interpretation. We identify a no-additive-interaction condition under which these quantities coincide, develop sensitivity bounds for departures from that condition, and propose an alternative model-robust estimand. This estimand equals the ATE when the variable is pre-exposure and an interventional direct effect when it is post-exposure. Moreover, within a natural class of outcome-free stochastic direct effects, it is the unique observed-data functional that remains causally interpretable under both structural roles of the focal variable. We derive an efficient influence function and a doubly robust estimator, yielding robustness at two levels: the estimand is model-robust across the two causal scenarios, and the estimator is doubly robust with respect to nuisance estimation. In simulations and in an NHANES application on elevated PFAS burden, kidney function, and uric acid, mediation-based analyses yielded materially different reported estimates.

\end{abstract}


\section{Introduction}

Direct effects are often used to ask whether an exposure changes an outcome through pathways that do not operate through a focal variable $W$. That question is central in many fields of science because it helps distinguish overall impact from pathway-specific impact and can guide mechanistic interpretation and intervention design \citep{robins1992identifiability,pearl2001direct,vanderweele2015explanation}. In many applications, the focal variable may be known to predict the outcome while its causal relationship with the exposure remains unclear. This can happen when exposure timing is coarse, when repeated processes are summarized by a single measurement, or when substantive experts disagree about whether the variable is antecedent to exposure or partly generated by it. Our motivating application concerns the direct effect of elevated per- and polyfluoroalkyl substance (PFAS) burden on uric acid relative to kidney function, summarized by estimated glomerular filtration rate (eGFR), in NHANES. Kidney function is structurally ambiguous in this setting. On the one hand, lower kidney function can increase measured serum PFAS concentrations through reduced renal clearance and is itself a strong determinant of serum uric acid, so eGFR may be viewed as a pre-exposure characteristic associated with both PFAS burden and uric acid \citep{kshirsagar2022environmental,dhingra2017reverse,moon2021pfas_kidney}. On the other hand, PFAS exposure has also been associated with lower kidney function, in which case eGFR may lie on the pathway from PFAS burden to uric acid \citep{lin2021pfas_kidney,niu2024kidney}. In such settings, a direct effect analysis faces an under-appreciated structural problem: the target parameter itself may change with the causal model adopted for the same observed data.

This difficulty is not merely semantic. If $W$ is a pre-exposure confounder, then there is no indirect pathway from the exposure $A$ to the outcome $Y$ through $W$. Therefore, the direct effect of $A$ on $Y$ relative to $W$ can be evaluated using parameters describing total causal effect, in particular the \emph{average treatment effect}. If $W$ is a post-exposure mediator, then the direct effect is usually defined through nested potential outcomes such as the \emph{natural direct effect}. These quantities generally average the same conditional exposure contrast over different distributions of $W$ and therefore need not agree. More importantly, if a practitioner applies a natural direct effect estimator when $W$ is actually pre-exposure, the estimator may converge to an observed-data functional that does not correspond to a causal direct effect. The resulting analysis may therefore be numerically stable yet substantively misleading.

In this paper, we formalize this challenge and show that it creates a genuine estimand problem. Our central contribution concerns model robustness, or identification robustness, rather than the more familiar estimator-level double robustness for a fixed estimand. First, we prove that there is generally no single observed-data functional that identifies the average treatment effect when $W$ is a confounder and the natural direct effect when $W$ is a mediator. Second, we show that those parameters coincide under a no-additive-interaction restriction and derive a sensitivity bound that quantifies departures from that restriction. Third, we propose a model-robust alternative estimand and show that it equals the average treatment effect under the confounder model and an interventional direct effect under the mediator model. Fourth, we show that this choice is canonical: within a natural class of outcome-free stochastic direct effects, the proposed parameter is the unique observed-data functional that coincides with the average treatment effect under the confounder model, and hence the unique member of this class that remains causally interpretable under both structural roles of $W$. Fifth, we derive a cross-fitted one-step estimator for this estimand and discuss hypothesis testing under structural ambiguity. Thus the proposed procedure is robust at two levels: the estimand is model-robust across the two causal scenarios, and the estimator is doubly robust with respect to nuisance estimation. Finally, we evaluate the method in a simulation study and in an NHANES application on elevated PFAS burden, kidney function, and uric acid. All proofs are provided in the Appendix.

\subsection{Related Work}
\label{sec:related}

The modern mediation literature was built on formal potential-outcome definitions of direct and indirect effects \citep{robins1992identifiability,pearl2001direct}. These ideas were developed further through identification theory, sensitivity analysis, and applied guidance in statistics  \citep{vanderweele2009conceptual,imai2010identification,vanderweele2015explanation}. In that literature, the natural direct effect has become a canonical estimand, but it typically relies on strong assumptions, including cross-world conditions or structural models that imply them \citep{imai2010identification,tchetgen2012semiparametric}. Recent reviews have also emphasized how demanding these assumptions can be in practice, especially when temporal ordering is unclear or only coarsely measured \citep{stuart2021assumptions,schuler2025practical}. Especially relevant here, \citet{georgeson2023sensitivity} study temporal bias in cross-sectional mediation and propose a sensitivity analysis for situations in which mediation is analyzed despite unresolved temporal ordering. Their focus differs from ours, however, because they address bias within mediation analysis once the variable is treated as a mediator, whereas we study ambiguity about whether the variable is pre-exposure or post-exposure in the first place.

A parallel literature introduced alternative direct and indirect effect definitions, such as interventional, stochastic, organic, and population-intervention effects, that remain meaningful when deterministic interventions on the mediator are difficult to interpret or when cross-world assumptions are undesirable \citep{hubbard2008population,geneletti2007identifying,vansteelandt2017interventional,diaz2020causal,hejazi2023nonparametric,lok2016defining}. Related work has studied population indirect effects under front-door-type conditions \citep{fulcher2020robust} and has examined the causal interpretation of randomized interventional analogues and when they coincide with or differ from natural effects \citep{miles2023causal,yu2024detecting}. These contributions clarify what can still be learned when natural effects are unavailable or unattractive and motivate interventional direct effects as meaningful causal alternatives. Our problem, however, is different from the usual motivation for these estimands. We do not assume that $W$ is definitely a mediator and then seek a weaker, more policy-relevant, or more easily identified direct effect. Instead, we allow ambiguity about whether $W$ is pre-exposure or post-exposure in the first place and ask whether a single observed-data functional can retain a direct-effect interpretation across those competing structural roles. To our knowledge, the literature has not explicitly characterized what can be learned from a single observed-data law when those two structural interpretations are both plausible, nor identified a canonical common target within a natural class of stochastic direct effects.

Our setting is also related to work on ambiguity about how intermediate variables should be handled in adjustment. \citet{inoue2020confounder} discussed this issue explicitly as a confounder-mediator dilemma for obesity/BMI in environmental epidemiology, emphasizing that in cross-sectional settings the same variable may plausibly be treated as either a confounder or a mediator. \citet{wang2017mistakenly} studied bias from mistakenly adjusting for a mediator when estimating a total effect. \citet{dyer2025} emphasized more generally that when confounder status is unclear, sensitivity analyses over variable inclusion are preferable to treating confounding as a purely descriptive notion. More closely related, \citet{takahashi2025mediating} studied variables that are partly mediator and partly confounder within the same dataset. These papers underscore the practical importance of structural ambiguity, but they do not characterize a common observed-data target that remains causally interpretable when the same variable may be either a confounder or a mediator.

Our paper also connects to work on identification-level robustness under a union of causal models. Related in spirit, but in a different setting, \citet{wang2018bounded} showed that in an instrumental variable problem the average treatment effect can be identified by the same average Wald functional under either of two no-interaction assumptions. Our contribution differs in that the ambiguity here is between confounding and mediation for an intermediate variable, and the goal is a direct effect that remains meaningful across those two competing causal scenarios.

Finally, the paper connects to semiparametric work on causal functionals and efficient estimation \citep{tsiatis2006semiparametric,bang2005doubly}. On the mediation side, \cite{tchetgen2012semiparametric} developed the influence function based estimation strategy, \cite{zheng2012targeted} developed targeted minimum loss-based estimation for the natural direct effect, and \citet{diaz2021nonparametric} and \citet{hejazi2023nonparametric} developed efficient machine learning-based estimators for interventional and stochastic interventional effects settings. We exploit the structure of our model-robust target to obtain an estimator that is easy to implement and does not require a different inferential strategy depending on whether $W$ is a mediator or a confounder. This estimation robustness is distinct from, and complements, the identification robustness that is the main focus of the paper.

\section{Model Description}
\label{sec:desc}

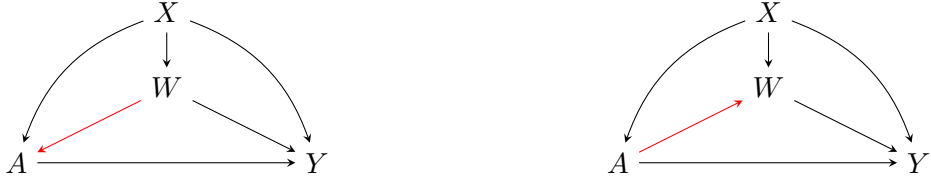
\begin{figure}[t]
\centering
		\tikzstyle{block} = [draw, circle, inner sep=2.5pt, fill=lightgray]
		\tikzstyle{input} = [coordinate]
		\tikzstyle{output} = [coordinate]
        \begin{tikzpicture}
            \tikzset{edge/.style = {->,> = latex'}}
            \node[] (a) at  (-2,0) {$A$};
            \node[] (x) at  (0,2) {$X$};
            \node[] (w) at  (0,1) {$W$};
            \node[] (y) at  (2,0) {$Y$};
            \node[] (a') at  (6,0) {$A$};
            \node[] (x') at  (8,2) {$X$};
            \node[] (w') at  (8,1) {$W$};
            \node[] (y') at  (10,0) {$Y$};                                   
            \draw[-stealth,red] (w) to (a);                                                         
            \draw[-stealth] (a) to (y);
            \draw[-stealth,bend left=-25] (x) to (a);                                                         
            \draw[-stealth] (x) to (w);
            \draw[-stealth,bend left=25] (x) to (y);
            \draw[-stealth] (w) to (y);
            \draw[-stealth,red] (a') to (w');                                                         
            \draw[-stealth] (a') to (y');
            \draw[-stealth,bend left=-25] (x') to (a');                                                         
            \draw[-stealth] (x') to (w');
            \draw[-stealth,bend left=25] (x') to (y');
            \draw[-stealth] (w') to (y');            
        \end{tikzpicture}
        \caption{Graphical representations of the causal models for the confounding (left) and mediation (right) scenarios.}
        \label{fig:MemMod}
\end{figure}

Let $O=(X,A,W,Y)$ denote the observed data, where $A\in\{0,1\}$ is the exposure, $Y$ is the outcome, $X$ denotes baseline covariates measured before $A$, and $W$ is a \emph{focal} variable measured before $Y$. We observe independent and identically distributed copies $O_1,\ldots,O_n\sim P$, where $P$ denotes the observed-data law. Throughout, $W$ is known to be causally relevant for $Y$, but its causal position relative to $A$ is ambiguous; see Figure \ref{fig:MemMod}.

We write
\[
Q(a,w,x):=\E_P(Y\mid A=a,W=w,X=x)
\]
and
\[
\Delta(w,x):=Q(1,w,x)-Q(0,w,x).
\]
The quantity $\Delta(w,x)$ is the observed conditional mean contrast for exposure at a fixed $(w,x)$. A central theme of the paper is that the same contrast $\Delta(w,x)$ is averaged over different distributions of $W$ depending on whether $W$ is treated as a confounder or as a mediator.

\begin{assumption}[Positivity]\label{ass:positivity}
For $P$-almost every $(x,w)$ in the support of $(X,W)$,
\[
0<\Prob(A=1\mid X=x,W=w)<1.
\]
Under the mediator model, we also assume
\[
0<\Prob(A=1\mid X=x)<1
\]
for $P$-almost every $x$.
\end{assumption}

\subsection{Confounder Model}

Under the confounder model, denoted $\mathcal{C}$, the variable $W$ is pre-exposure and may affect both $A$ and $Y$. Potential outcomes are denoted $Y^{(1)}$ and $Y^{(0)}$.

\begin{assumption}[Confounder model consistency and exchangeability]\label{ass:confounder}
Under $\mathcal{C}$:
\begin{enumerate}[label=(\alph*)]
\item if $A=a$, then $Y=Y^{(a)}$ almost surely;
\item for $a\in\{0,1\}$, $Y^{(a)}\indep A\mid (X,W)$.
\end{enumerate}
\end{assumption}

If $W$ is genuinely pre-exposure, then it cannot carry any of the effect of $A$ to $Y$. Relative to $W$, the direct effect is therefore the total effect, commonly summarized by the average treatment effect.

\begin{definition}[Average treatment effect]
The average treatment effect is defined as
\[
\theta_{\mathrm{ATE}}:=\E\!\left\{Y^{(1)}-Y^{(0)}\right\}.
\]
\end{definition}

\begin{proposition}
\label{prop:confounder-id}
Under Assumptions~\ref{ass:positivity} and \ref{ass:confounder},
\[
\theta_{\mathrm{ATE}}
=
\E_P\!\left\{\Delta(W,X)\right\}
=
\E_P\!\left[\int \Delta(w,X)\,\dd F_{W\mid X}(w\mid X)\right].
\]
\end{proposition}

When $W$ is a confounder, the causal effect of $A$ is identified after conditioning on $(X,W)$. The relevant averaging distribution is therefore the observed distribution of $W$ given $X$, because this is the covariate distribution in the target population.

\subsection{Mediator Model}

Under the mediator model, denoted $\mathcal{M}$, the variable $W$ is post-exposure. We write $W^{(a)}$ for the potential mediator under exposure level $a$, and $Y^{(a,w)}$ for the potential outcome if exposure were set to $a$ and the mediator were set to $w$. We also use the composition relation $Y^{(a)} = Y^{(a,W^{(a)})}$.

\begin{assumption}[Mediator model consistency and exchangeability]\label{ass:mediator}
Under $\mathcal{M}$:
\begin{enumerate}[label=(\alph*)]
\item if $A=a$, then $W=W^{(a)}$ almost surely;
\item if $A=a$ and $W=w$, then $Y=Y^{(a,w)}$ almost surely;
\item for all $a,a'\in\{0,1\}$ and all $w$,
\[
\bigl(Y^{(a,w)},W^{(a')}\bigr)\indep A\mid X;
\]
\item for all $a\in\{0,1\}$ and all $w$,
\[
Y^{(a,w)}\indep W\mid (A=a,X).
\]
\end{enumerate}
\end{assumption}

\begin{assumption}[Cross-world condition]\label{ass:crossworld}
For all $a,a'\in\{0,1\}$ and all $w$,
\[
Y^{(a,w)}\indep W^{(a')}\mid X.
\]
\end{assumption}

Assumption~\ref{ass:crossworld} is a standard sufficient condition for identifying natural direct effects.

\begin{definition}[Natural direct effect]
Under $\mathcal{M}$, the natural direct effect is
\[
\theta_{\mathrm{NDE}}
:=
\E\!\left\{Y^{(1,W^{(0)})}-Y^{(0,W^{(0)})}\right\}.
\]
\end{definition}

\begin{proposition}
\label{prop:nde-id}
Under Assumptions~\ref{ass:positivity}, \ref{ass:mediator}, and \ref{ass:crossworld},
\[
\theta_{\mathrm{NDE}}
=
\E_P\!\left[\int \Delta(w,X)\,\dd F_{W\mid A=0,X}(w\mid 0,X)\right].
\]
\end{proposition}

The natural direct effect keeps the mediator at the distribution it would have taken under $A=0$. Consequently, the same conditional exposure contrast $\Delta(w,x)$ is averaged over $F_{W\mid A=0,X=x}$ rather than over $F_{W\mid X=x}$.

\subsection{No Common Identified Functional Under Confounder--Mediator Ambiguity}

Propositions~\ref{prop:confounder-id} and \ref{prop:nde-id} reveal a central difficulty: under the confounder model, the relevant averaging law is $F_{W\mid X}$. Under the mediator model, the natural direct effect averages over $F_{W\mid A=0,X}$. These distributions need not coincide. The next theorem shows that, in general, there is no single observed-data functional that identifies the direct effect under the confounder model in one causal scenario and the natural direct effect under the mediator model in the other.

\begin{theorem}
\label{thm:no-common}
There does not exist a map $T$ that assigns a real number to each observed-data law $P$ such that
\begin{enumerate}[label=(\roman*)]
\item $T(P)=\theta_{\mathrm{ATE}}$ for every observed-data law induced by a causal data-generating process satisfying Assumptions \ref{ass:positivity} and \ref{ass:confounder}; and
\item $T(P)=\theta_{\mathrm{NDE}}$ for every observed-data law induced by a causal data-generating process satisfying Assumptions \ref{ass:positivity}, \ref{ass:mediator}, and \ref{ass:crossworld}.
\end{enumerate}
\end{theorem}

The proof of Theorem~\ref{thm:no-common} demonstrates that the \emph{same} observed-data law can arise from two different causal scenarios that imply different direct effect parameters.  Structural ambiguity therefore becomes estimand ambiguity.

\section{Coincidence of the Models Under No Additive Interaction}
\label{sec:no-interaction}

The impossibility in Theorem~\ref{thm:no-common} disappears under an effect-homogeneity restriction: if the conditional mean effect of $A$ on $Y$ does not depend on $W$ once $X$ is fixed, then the relevant parameters coincide.

\begin{assumption}[No additive interaction between $A$ and $W$]\label{ass:no-interaction}
There exists a measurable function $\delta(x)$ such that, for $P$-almost every $(w,x)$,
\[
\Delta(w,x)=\delta(x).
\]
\end{assumption}

A broad class of data-generating processes satisfying Assumption~\ref{ass:no-interaction} is obtained when the outcome equation is additive in $A$ on the mean scale. For example, under a mediator model, suppose
\[
A = h(X,U_A), \qquad
W = r(A,X,U_W), \qquad
Y = \eta(W,X,U_Y) + A\,\delta(X),
\]
where $h$, $r$, and $\eta$ are measurable functions, $\delta$ is the function in Assumption~\ref{ass:no-interaction}, and $U_A$, $U_W$, and $U_Y$ are mutually independent and independent of $X$. Then
\[
Q(a,w,x)=\E(Y\mid A=a,W=w,X=x)
       = \E\{\eta(w,x,U_Y)\}+a\,\delta(x),
\]
so
\[
\Delta(w,x)=Q(1,w,x)-Q(0,w,x)=\delta(x),
\]
which does not depend on $w$. Hence Assumption~\ref{ass:no-interaction} holds. The key feature is that $W$ may affect $Y$ through $\eta(W,X,U_Y)$, but the additive effect of $A$ on $Y$ does not vary with $W$. An analogous confounder-model example is obtained by instead taking $W=r(X,U_W)$ and $A=h(W,X,U_A)$, while keeping the same outcome equation for $Y$.

In what follows, we use the notation
\[
\psi(P):=\E_P\{\Delta(W,X)\}.
\]

\begin{theorem}
\label{thm:no-interaction}
Suppose Assumption~\ref{ass:no-interaction} holds. Then:
\begin{enumerate}[label=(\roman*)]
\item under Assumptions~\ref{ass:positivity} and \ref{ass:confounder},
\[
\psi(P)=\theta_{\mathrm{ATE}}=\E_P\!\left\{\delta(X)\right\};
\]
\item under Assumptions~\ref{ass:positivity}, \ref{ass:mediator}, and \ref{ass:crossworld},
\[
\psi(P)=\theta_{\mathrm{NDE}}=\E_P\!\left\{\delta(X)\right\}.
\]
\end{enumerate}
\end{theorem}

Intuitively, if $\Delta(w,x)$ does not vary with $w$, then it does not matter which distribution of $W$ is used for averaging. Theorem~\ref{thm:no-interaction} formalizes the idea that interaction between $A$ and $W$ is what makes direct effect definitions diverge on the additive mean scale.

\subsection{An Exact Decomposition and Sensitivity Bound}

Theorem~\ref{thm:no-interaction} suggests a useful way to quantify how far the natural direct effect may be from the parameter $\psi(P)$ when Assumption~\ref{ass:no-interaction} fails. Under the mediator model, write
\[
\pi(x):=\Prob(A=1\mid X=x).
\]
The next result gives an exact decomposition of the discrepancy and then a sensitivity bound.

For each $x$, define the oscillation of $\Delta(\cdot,x)$ by
\[
\operatorname{rng}_{\Delta}(x)
:=
\operatorname*{ess\,sup}_{w}\Delta(w,x)
-
\operatorname*{ess\,inf}_{w}\Delta(w,x),
\]
and define the conditional total-variation distance
\[
\TV(x)
:=
\sup_{B\subseteq\mathcal{W}}
\left|
F_{W\mid X}(B\mid x)-F_{W\mid A=0,X}(B\mid 0,x)
\right|,
\]
where $\mathcal{W}$ denotes the support of $W$.

\begin{proposition}
\label{prop:sensitivity}
Under Assumptions~\ref{ass:positivity}, \ref{ass:mediator}, and \ref{ass:crossworld},
\begin{equation}
\label{eq:gap-conditional}
\psi(P)-\theta_{\mathrm{NDE}}
=
\E_P\!\left[
\pi(X)\Bigl(
\E_P\{\Delta(W,X)\mid A=1,X\}
-
\E_P\{\Delta(W,X)\mid A=0,X\}
\Bigr)
\right].
\end{equation}
Moreover,
\begin{equation}
\label{eq:sens}
\left|\theta_{\mathrm{NDE}}-\psi(P)\right|
\le
\E_P\!\left\{\operatorname{rng}_{\Delta}(X)\,\TV(X)\right\}.
\end{equation}
In particular, if $\operatorname{rng}_{\Delta}(X)\le \Gamma$ almost surely for some $\Gamma\ge 0$, then
\[
\theta_{\mathrm{NDE}}
\in
\left[
\psi(P)-\Gamma\,\E_P\{\TV(X)\},
\;
\psi(P)+\Gamma\,\E_P\{\TV(X)\}
\right].
\]
If $W$ is binary, then \eqref{eq:gap-conditional} simplifies to
\[
\psi(P)-\theta_{\mathrm{NDE}}
=
\E_P\!\left[
\bigl\{\Delta(1,X)-\Delta(0,X)\bigr\}
\Bigl\{
\Prob(W=1\mid X)-\Prob(W=1\mid A=0,X)
\Bigr\}
\right].
\]
\end{proposition}

Equation~\eqref{eq:gap-conditional} shows that the discrepancy is driven by treatment-induced changes in the conditional distribution of $W$ given $X$, and only to the extent that those changes alter the conditional average of $\Delta(w,X)$. The bound in \eqref{eq:sens} has two ingredients. The term $\operatorname{rng}_{\Delta}(X)$ measures how strongly the conditional effect of $A$ varies with $W$, and $\TV(X)$ measures how different the two relevant conditional distributions of $W$ are. If either ingredient is small, then the natural direct effect and the parameter $\psi(P)$ must be close.

\section{A Model-Robust Direct Effect}
\label{sec:psi}

We now define a direct effect parameter that remains meaningful under both causal scenarios. Let $G$ be a \emph{stochastic intervention} variable such that, conditional on $X=x$, $G$ has distribution $F_{W\mid X}(\cdot\mid x)$ and is independent of the potential outcomes $\{Y^{(a,w)}:a\in\{0,1\},w\in\mathcal{W}\}$ given $X=x$.

\begin{definition}[Interventional direct effect]
Under the mediator model, the interventional direct effect is defined as
\[
\theta_{\mathrm{IDE}}
:=
\E\!\left\{Y^{(1,G)}-Y^{(0,G)}\right\}.
\]
\end{definition}

This parameter compares the two exposure levels while drawing the mediator from its \emph{observed} conditional distribution given $X$. It is an \emph{interventional} direct effect because it references a stochastic intervention on the mediator distribution rather than the nested counterfactual $W^{(0)}$.

\begin{theorem}[The model-robust direct effect]
\label{thm:psi}
Recall the definition
\[
\psi(P):=\E_P\!\left\{\Delta(W,X)\right\}.
\]
\begin{enumerate}[label=(\roman*)]
\item under the confounder model $\mathcal{C}$ and Assumptions~\ref{ass:positivity} and \ref{ass:confounder},
\[
\psi(P)=\theta_{\mathrm{ATE}}=\E\!\left\{Y^{(1)}-Y^{(0)}\right\};
\]
\item under the mediator model $\mathcal{M}$ and Assumptions~\ref{ass:positivity} and \ref{ass:mediator},
\[
\psi(P)=\theta_{\mathrm{IDE}}=\E\!\left\{Y^{(1,G)}-Y^{(0,G)}\right\}.
\]
\end{enumerate}
\end{theorem}

In view of Theorem~\ref{thm:psi}, we refer to the parameter $\psi(P)$ as the \emph{model-robust direct effect}. Note that we do \emph{not} require the no-interaction assumption (Assumption~\ref{ass:no-interaction}) for $\psi(P)$ to be model-robust.

Theorem~\ref{thm:psi} uses the same observed-data functional under both models, but its causal interpretation changes with the structural role of $W$, while both interpretations retain a direct effect meaning. If $W$ is pre-exposure, $\psi(P)$ is simply the total effect because there is no mediated pathway through $W$. If $W$ is post-exposure, $\psi(P)$ is an interventional direct effect obtained by averaging over the observed conditional distribution of $W$ given $X$. This is why $\psi(P)$ is model-robust: the statistical target does not depend on whether the arrow between $A$ and $W$ points forward or backward. This kind of robustness is distinct from the usual estimator-level double robustness of augmented inverse-probability weighting and related semiparametric estimators \citep{bang2005doubly}; rather, it is closer to identification under a union of causal models. A close example is \citet{wang2018bounded}, who showed in an instrumental variable setting that the average treatment effect is identified by the same average Wald functional under either of two distinct no-interaction assumptions. To the best of our knowledge, the present result appears to be the first to obtain identification-level robustness specifically for ambiguity about whether a variable is a confounder or a mediator.

\paragraph{Practical attractiveness of $\psi(P)$.}

The parameter $\psi(P)$ has at least three practical advantages. First, it is identified by the same observed-data functional regardless of whether $W$ is a confounder or a mediator. A practitioner does not need to commit to a fragile structural choice before estimating the target parameter. Second, $\psi(P)$ avoids cross-world counterfactuals. The natural direct effect requires the nested quantity $Y^{(1,W^{(0)})}$, whereas $\psi(P)$ can be interpreted under the mediator model entirely through a stochastic intervention. Third, $\psi(P)$ is easy to estimate: influence function and doubly robust techniques can be used directly, as described in Subsection \ref{sec:IF}.

\subsection{Uniqueness}

Theorem~\ref{thm:psi} shows that $\psi(P)$ remains causally interpretable under both structural roles of $W$. The next result shows that, within a natural class of stochastic direct effects, this choice is also unique.

For any rule $P\mapsto \nu_P(\cdot\mid x)$ assigning to each observed-data law $P$ a conditional distribution on $W$ given $X=x$, define
\[
\Psi_\nu(P)
:=
\E_P\!\left[\int \Delta(w,X)\,\dd \nu_P(w\mid X)\right].
\]
We call the rule \emph{outcome-free} if it depends on $P$ only through the marginal law of $(X,A,W)$; equivalently, if $P_{XAW}=P'_{XAW}$, then $\nu_P(\cdot\mid x)=\nu_{P'}(\cdot\mid x)$ for $P_X$-almost every $x$. Under the mediator model, $\Psi_\nu(P)$ is the stochastic direct effect obtained by drawing the mediator from $\nu_P(\cdot\mid X)$ independently of the potential outcomes given $X$, by the same conditioning argument used in the proof of Theorem~\ref{thm:psi}.

\begin{theorem}[Uniqueness of $\psi(P)$]
\label{thm:psi-unique}
Let $P\mapsto \nu_P(\cdot\mid x)$ be outcome-free, and let $\Psi_\nu(P)$ be defined above. If
\[
\Psi_\nu(P)=\theta_{\mathrm{ATE}}
\]
for every confounder model law $P$ satisfying Assumptions~\ref{ass:positivity} and~\ref{ass:confounder}, then, for every observed-data law $P$ satisfying Assumption~\ref{ass:positivity},
\[
\nu_P(\cdot\mid X)=F_{W\mid X}(\cdot\mid X)
\qquad P\text{-almost surely}.
\]
Consequently,
\[
\Psi_\nu(P)=\E_P\!\left\{\Delta(W,X)\right\}=\psi(P).
\]
Hence $\psi(P)$ is the unique functional of the form
\[
P\mapsto \E_P\!\left[\int \Delta(w,X)\,\dd \nu_P(w\mid X)\right],
\]
with $\nu_P$ outcome-free, that coincides with the average treatment effect under the confounder model.
\end{theorem}

Combined with Theorem~\ref{thm:psi}, Theorem~\ref{thm:psi-unique} shows that $\psi(P)$ is the unique outcome-free stochastic direct effect that equals the average treatment effect under $\mathcal{C}$ and an interventional direct effect under $\mathcal{M}$.

\begin{remark}	
This outcome-free restriction is the natural analogue of how stochastic or interventional mediator distributions are usually specified: the intervention law describes how $W$ is to be drawn given baseline information and is determined from the law of $(X,A,W)$ or by an externally imposed design, not from the conditional law of $Y$ given $(A,W,X)$ \citep{hubbard2008population,vansteelandt2017interventional,diaz2020causal}. Allowing $\nu_P$ to depend on the outcome regression would blur the distinction between defining the estimand and estimating it, because the intervention rule could then be tuned using knowledge of $\Delta(w,x)$ itself. In that enlarged class, uniqueness would be vacuous: one could manufacture many law-specific choices of $\nu_P$ that reproduce $\theta_{\mathrm{ATE}}$ by construction. Restricting attention to outcome-free rules therefore isolates stochastic direct effects whose intervention component is specified independently of the outcome model, so Theorem~\ref{thm:psi-unique} says that among substantively interpretable candidates, $\psi(P)$ is the only one that agrees with the confounder model direct effect for every observed-data law.
\end{remark}

\subsection{Efficient Influence Function and One-Step Estimation}
\label{sec:IF}

Recall the definition $Q(a,w,x):=\E_P(Y\mid A=a,W=w, X=x)$, and define
\[
g(w,x):=\Prob(A=1\mid W=w, X=x),
\]
for all $a,w,x$. Note that
\[
\psi(P)=\E_P\!\left\{Q(1,W,X)-Q(0,W,X)\right\}.
\]

\begin{theorem}[Efficient influence function]\label{thm:eif}
Under Assumption~\ref{ass:positivity} and the nonparametric model for $P$, an efficient influence function for $\psi(P)$ is
\begin{align}
\varphi_{\psi}(O)
=
&\frac{\one\{A=1\}}{g(W,X)}\left\{Y-Q(1,W,X)\right\}
-\frac{\one\{A=0\}}{1-g(W,X)}\left\{Y-Q(0,W,X)\right\}
\notag\\[3pt]
&\qquad
+\Delta(W,X)-\psi(P).
\label{eq:eif}
\end{align}
\end{theorem}

The regression $Q(a,w,x)$ is the conditional mean of $Y$ given $(A,W,X)$ in both causal models. The function
\(
g(w,x)=\Prob(A=1\mid W=w,X=x)
\)
is also an observed-data object. Under the confounder model it is a propensity score conditional on confounders. Under the mediator model it is not a treatment-assignment mechanism in a causal sense, because $W$ is post-exposure; nevertheless, it remains a valid nuisance function for the \emph{statistical} parameter $\psi(P)$. This is another sense in which the estimation problem is model-robust.

We estimate $\psi(P)$ with a $K$-fold cross-fitted one-step estimator \citep{chernozhukov18}. Let $\mathcal{I}_1,\ldots,\mathcal{I}_K$ be a random partition of $\{1,\ldots,n\}$ into $K$ approximately equal folds, with $K$ fixed, and let $\mathcal{I}_{-k}:=\{1,\ldots,n\}\setminus \mathcal{I}_k$. For each fold $k$, use only the training sample $\{O_j:j\in\mathcal{I}_{-k}\}$ to construct nuisance estimators $\widehat Q^{(-k)}$ and $\widehat g^{(-k)}$ of $Q$ and $g$. Then, for each observation $i\in\mathcal{I}_k$, compute the held-out score contribution
\begin{align*}
\widehat\phi^{\mathrm{cf}}_i
=
&\frac{\one\{A_i=1\}}{\widehat g^{(-k)}(W_i,X_i)}
\left\{Y_i-\widehat Q^{(-k)}(1,W_i,X_i)\right\}
-
\frac{\one\{A_i=0\}}{1-\widehat g^{(-k)}(W_i,X_i)}
\left\{Y_i-\widehat Q^{(-k)}(0,W_i,X_i)\right\}
\\[3pt]
&\qquad\qquad
+\widehat Q^{(-k)}(1,W_i,X_i)-\widehat Q^{(-k)}(0,W_i,X_i).
\end{align*}
The cross-fitted one-step estimator is
\begin{equation}
\widehat\psi_{\mathrm{cf}}
=
\frac{1}{n}\sum_{k=1}^K \sum_{i\in\mathcal{I}_k}\widehat\phi^{\mathrm{cf}}_i.
\label{eq:psi-est}
\end{equation}
Cross-fitting means that each observation is evaluated using nuisance estimators obtained from a subsample that does not include that observation. This out-of-fold construction reduces overfitting bias and allows flexible machine learning methods to be used for nuisance estimation under weaker empirical-process conditions than full-sample plug-in fitting \citep{chernozhukov18}.

\begin{theorem}[Double robustness and asymptotic normality]
\label{thm:dr}
Fix $K\ge 2$. Let $\widehat\psi_{\mathrm{cf}}$ be the $K$-fold cross-fitted one-step estimator defined in \eqref{eq:psi-est}, and let
\(
g_0(w,x):=\Prob(A=1\mid W=w,X=x).
\)
Assume $Y\in L_2(P)$, $Q(a,\cdot)\in L_2(P)$ for $a=0,1$, and that there exists $\varepsilon_0>0$ such that
\(
\varepsilon_0 \le g_0(W,X)\le 1-\varepsilon_0
\)
almost surely. Assume also that there exists $\varepsilon>0$ such that, with probability tending to $1$,
\(
\varepsilon \le \widehat g^{(-k)}(w,x)\le 1-\varepsilon
\)
for all $(w,x)$ and all $k=1,\ldots,K$. Then:
\begin{enumerate}[label=(\roman*)]
\item If either
\[
\max_{1\le k\le K}\sum_{a=0}^1
\left\|\widehat Q^{(-k)}(a,\cdot)-Q(a,\cdot)\right\|_{L_2(P)}
=
o_P(1),
\]
or
\[
\max_{1\le k\le K}
\left\|\widehat g^{(-k)}-g_0\right\|_{L_2(P)}
=
o_P(1)
\]
together with
\[
\max_{1\le k\le K}\sum_{a=0}^1
\left\|\widehat Q^{(-k)}(a,\cdot)\right\|_{L_2(P)}
=
O_P(1),
\]
then
\[
\widehat\psi_{\mathrm{cf}}\overset{P}{\longrightarrow}\psi(P).
\]

\item If
\[
\max_{1\le k\le K}\sum_{a=0}^1
\left\|\widehat Q^{(-k)}(a,\cdot)-Q(a,\cdot)\right\|_{L_2(P)}
=
o_P(1),
\qquad
\max_{1\le k\le K}
\left\|\widehat g^{(-k)}-g_0\right\|_{L_2(P)}
=
o_P(1),
\]
and
\[
\max_{1\le k\le K}
\left\|\widehat g^{(-k)}-g_0\right\|_{L_2(P)}
\,
\max_{1\le k\le K}\sum_{a=0}^1
\left\|\widehat Q^{(-k)}(a,\cdot)-Q(a,\cdot)\right\|_{L_2(P)}
=
o_P(n^{-1/2}),
\]
then
\[
\sqrt{n}\left\{\widehat\psi_{\mathrm{cf}}-\psi(P)\right\}
\overset{d}{\longrightarrow}
N\!\left(0,\Var_P\{\varphi_{\psi}(O)\}\right).
\]
\end{enumerate}

Define the centered fold-specific estimated influence values by
\(
\widehat\varphi^{\mathrm{cf}}_i
:=
\widehat\phi^{\mathrm{cf}}_i-\widehat\psi_{\mathrm{cf}},
\)
for $i=1,\ldots,n$, where $\widehat\phi_i^{\mathrm{cf}}$ is the held-out score contribution defined above. An asymptotically valid standard error is
\[
\widehat{\mathrm{SE}}(\widehat\psi_{\mathrm{cf}})
=
\sqrt{
\frac{1}{n(n-1)}
\sum_{i=1}^n
\left(\widehat\varphi^{\mathrm{cf}}_i\right)^2
}.
\]
\end{theorem}

The estimator has the familiar \emph{estimation} doubly robust form: it combines an outcome regression with a weighting correction, and cross-fitting evaluates each score contribution out of sample. If either nuisance component is consistently estimated while the other remains well behaved, the second-order bias vanishes. If both nuisance components are estimated well enough, the estimator is asymptotically linear with influence function \eqref{eq:eif}, which yields Wald-type confidence intervals and tests. Thus the proposed procedure is robust at two levels: the estimand is model-robust across the two causal scenarios, and the estimator is doubly robust with respect to nuisance estimation.

\subsection{Testing the Null of No Direct Effect Under Structural Ambiguity}
\label{sec:testing}

If one insists on a model-specific definition of the direct effect---namely, $\theta_{\mathrm{ATE}}$ under $\mathcal{C}$ and $\theta_{\mathrm{NDE}}$ under $\mathcal{M}$---then the null of no direct effect under ambiguity about the structural role of $W$ becomes the composite union null
\[
H_0^{\cup}:
\quad
(\mathcal{C}\ \text{and}\ \theta_{\mathrm{ATE}}=0)
\ \ \cup\ \ 
(\mathcal{M}\ \text{and}\ \theta_{\mathrm{NDE}}=0).
\]
Compared with an ordinary single-parameter test, this null depends on which structural model generated the data.

Suppose $p_C$ is a valid $p$-value for testing $\theta_{\mathrm{ATE}}=0$ under $\mathcal{C}$ and $p_M$ is a valid $p$-value for testing $\theta_{\mathrm{NDE}}=0$ under $\mathcal{M}$; for example, these may be obtained from standard semiparametric estimators for the average treatment effect \citep{Robins01091994} and the natural direct effect \citep{tchetgen2012semiparametric}. Define
\[
p_{\max}:=\max\{p_C,p_M\}.
\]

\begin{proposition}
\label{prop:union-test}
The test that rejects $H_0^{\cup}$ whenever $p_{\max}<\alpha$ has size at most $\alpha$. Equivalently, one rejects only when both model-specific tests reject at level $\alpha$.
\end{proposition}

The approach above is an intersection--union argument. If the data were generated under the confounder model, then only the validity of $p_C$ matters; if they were generated under the mediator model, then only the validity of $p_M$ matters. Requiring both tests to reject is therefore conservative under the union null. By contrast, if one adopts the model-robust parameter $\psi(P)$, then there is a single null
\[
H_0^{\psi}:\quad \psi(P)=0,
\]
with a single Wald test based on Theorem~\ref{thm:dr}. This is another practical advantage of the proposed target: ambiguity about the structural role of $W$ does not change the formal testing problem.

\section{Simulation Study}\label{sec:sim}

The simulation study had three aims: to assess the finite-sample performance of the proposed estimator $\widehat\psi_{\mathrm{cf}}$, to evaluate confidence interval coverage, and to illustrate the estimand mismatch that arises when a mediation formula estimator is used when the structural role of $W$ is ambiguous.

We considered three data-generating cases, each with $X\sim N(0,1)$ and independent $N(0,1)$ outcome noise.

\begin{enumerate}[label=\textbf{Case \arabic*.}, leftmargin=2.35cm]
\item \textbf{Confounder with interaction.} The variable $W$ was generated before $A$ and affected both $A$ and $Y$:
\begin{align*}
&\Prob(W=1\mid X)=\operatorname{expit}(-0.2+0.6X),\\
&\Prob(A=1\mid W,X)=\operatorname{expit}(-0.1+0.5X+1.3W),\\
&Y=0.5+0.8A+0.9W+1.0\,AW+0.4X+\varepsilon.
\end{align*}

\item \textbf{Mediator without interaction.} The variable $W$ was generated after $A$, and there was no exposure--$W$ interaction in the outcome:
\begin{align*}
&\Prob(A=1\mid X)=\operatorname{expit}(-0.1+0.5X),\\
&\Prob(W=1\mid A,X)=\operatorname{expit}(-0.4+1.4A+0.6X),\\
&Y=0.5+0.8A+0.9W+0.4X+\varepsilon.
\end{align*}

\item \textbf{Mediator with interaction.} The data-generating mechanism was the same as in Case~2 except that the outcome included an interaction term:
\[
Y=0.5+0.8A+0.9W+1.0\,AW+0.4X+\varepsilon.
\]
\end{enumerate}

For each case and each sample size $n\in\{500,2000\}$, we generated 500 Monte Carlo samples. The proposed estimator was the 5-fold cross-fitted one-step estimator $\widehat\psi_{\mathrm{cf}}$. Within each training fold, the nuisance functions $Q(a,w,x)=\E(Y\mid A=a,W=w,X=x)$ and $g(w,x)=\Prob(A=1\mid W=w,X=x)$ were estimated with parametric working models. Specifically, the outcome regression was fit by linear regression with terms $(1,A,W,AW,X)$, and $g$ was fit by logistic regression of $A$ on $(1,W,X)$. The outcome model is correctly specified in all three cases. The held-out score contributions were then combined according to \eqref{eq:psi-est}, and Wald intervals were constructed using the cross-fitted influence-function variance estimator from Theorem~\ref{thm:dr}.

As a contrast, we also computed a parametric mediation-formula plug-in comparator targeting
\[
\lambda(P):=\E_P\!\left[\int \Delta(w,X)\,\dd F_{W\mid A=0,X}(w\mid 0,X)\right].
\]
In Cases~2 and~3, the working models used for this comparator are correctly specified, so under the mediator model it targets $\theta_{\mathrm{NDE}}$. In Case~1, it is included as a working-model benchmark and does not in general exactly recover $\lambda(P)$. In Case~1, however, $\lambda(P)$ is simply the observed-data functional targeted by the comparator and does not have a clear causal interpretation, because $W$ is generated as a pre-exposure confounder. Bias and root mean squared error for the comparator were evaluated relative to $\psi(P)$, because the point of the comparison is to show when the mediation-formula target diverges from the model-robust direct effect.

The true values were
\(
(\psi,\lambda)=(1.254,1.076)
\)
in Case~1,
\(
(\psi,\lambda)=(0.800,0.800)
\)
in Case~2, and
\(
(\psi,\lambda)=(1.353,1.209)
\)
in Case~3. Thus Cases~1 and 3 were chosen so that the proposed parameter and the comparator target differ for substantive reasons, whereas Case~2 satisfies the no-interaction condition in Theorem~\ref{thm:no-interaction}.

\subsection{Results}

Table~\ref{tab:sim} shows that $\widehat\psi_{\mathrm{cf}}$ was essentially unbiased in all three cases and at both sample sizes. Its absolute bias never exceeded 0.003, its empirical root mean squared error decreased substantially as $n$ increased, and the empirical coverage of nominal 95\% Wald intervals ranged from 0.946 to 0.962.

\begin{table}[t!]
\centering
\caption{Finite-sample performance over 500 Monte Carlo samples. The table reports the Monte Carlo mean, bias, RMSE, and empirical 95\% confidence interval coverage for $\widehat\psi_{\mathrm{cf}}$, together with the Monte Carlo mean of the mediation formula comparator and its bias relative to $\psi(P)$.}
\label{tab:sim}
\resizebox{\textwidth}{!}{%
\begin{tabular}{lrrrrrrr}
\toprule
Case & $n$ & Mean $\widehat\psi_{\mathrm{cf}}$ & Bias & RMSE & 95\% coverage & Mean comparator & Comparator bias to $\psi(P)$ \\
\midrule
Case 1: Confounder with interaction &  500 & 1.255 & 0.001 & 0.109 & 0.946 & 1.073 & -0.181 \\
Case 1: Confounder with interaction & 2000 & 1.254 & -0.000 & 0.054 & 0.948 & 1.075 & -0.179 \\
Case 2: Mediator without interaction &  500 & 0.803 & 0.003 & 0.096 & 0.962 & 0.803 & 0.003 \\
Case 2: Mediator without interaction & 2000 & 0.801 & 0.001 & 0.049 & 0.960 & 0.802 & 0.002 \\
Case 3: Mediator with interaction &  500 & 1.356 & 0.003 & 0.101 & 0.950 & 1.212 & -0.141 \\
Case 3: Mediator with interaction & 2000 & 1.355 & 0.002 & 0.049 & 0.954 & 1.212 & -0.141 \\
\bottomrule
\end{tabular}}
\end{table}

The mediation formula comparator behaved as the theory predicts: in Case~2, where there was no additive interaction, it agreed with $\psi(P)$ up to Monte Carlo error. In Case~1, where $W$ was actually a confounder, it converged to a value about 0.18 below $\psi(P)$. In Case~3, where $W$ was a mediator but exposure--$W$ interaction was present, it converged to a value about 0.14--0.15 below $\psi(P)$. Importantly, this discrepancy is not caused by numerical instability: the comparator was centered near its own large-sample limit in all three cases. In Cases~2 and 3, that limit equals the natural direct effect. In Case~1, however, the same limit does not have a clear causal interpretation because the data were generated under a confounder structure. The discrepancy in Cases~1 and 3 therefore reflects estimand mismatch rather than estimator failure.

Figure~\ref{fig:sim-box} shows the full Monte Carlo distributions. The solid horizontal line in each panel marks the true value of $\psi(P)$, while the dashed line marks $\lambda(P)$, the large-sample limit of the mediation formula comparator. In Cases~2 and 3, $\lambda(P)$ equals the natural direct effect. In Case~1, it does not have a clear causal interpretation. The proposed estimator tracks the solid line across all cases. The comparator tracks the dashed line and coincides with the solid line only in the no-interaction case. Figure~\ref{fig:sim-summary} summarizes empirical bias relative to $\psi(P)$ and confidence interval coverage for $\widehat\psi_{\mathrm{cf}}$. The main point is not only that $\widehat\psi_{\mathrm{cf}}$ performs well statistically, but also that it continues to target the same scientifically interpretable quantity across the two competing causal stories for $W$.

\begin{figure}[t!]
\centering
\includegraphics[width=\textwidth]{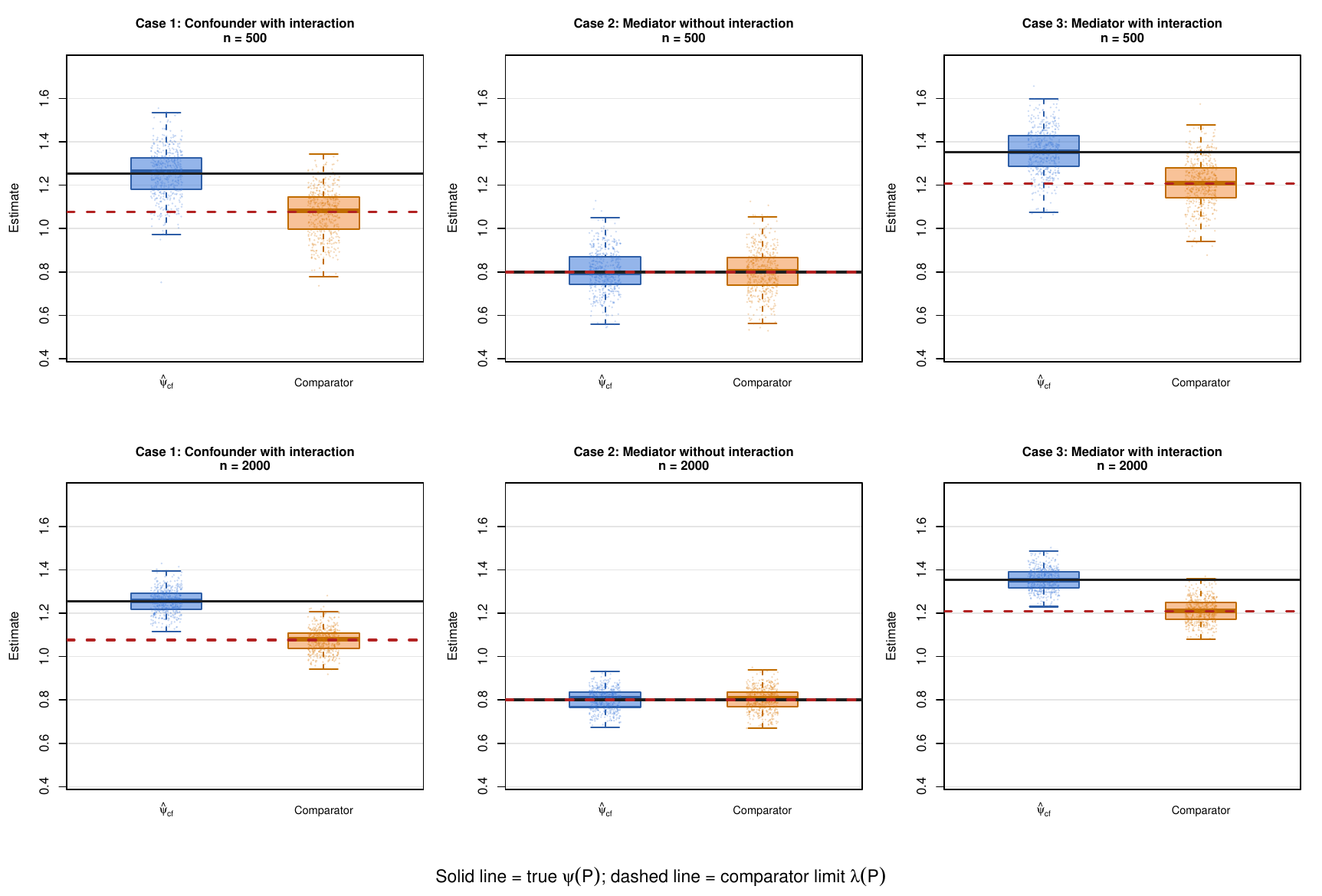}
\caption{Monte Carlo distributions of the cross-fitted estimator $\widehat\psi_{\mathrm{cf}}$ and the mediation formula comparator. The solid line is the true value of $\psi(P)$; the dashed line is $\lambda(P)$, the large-sample limit of the comparator. In Cases~2 and 3, $\lambda(P)=\theta_{\mathrm{NDE}}$.}
\label{fig:sim-box}
\end{figure}

\begin{figure}[t!]
\centering
\includegraphics[width=\textwidth]{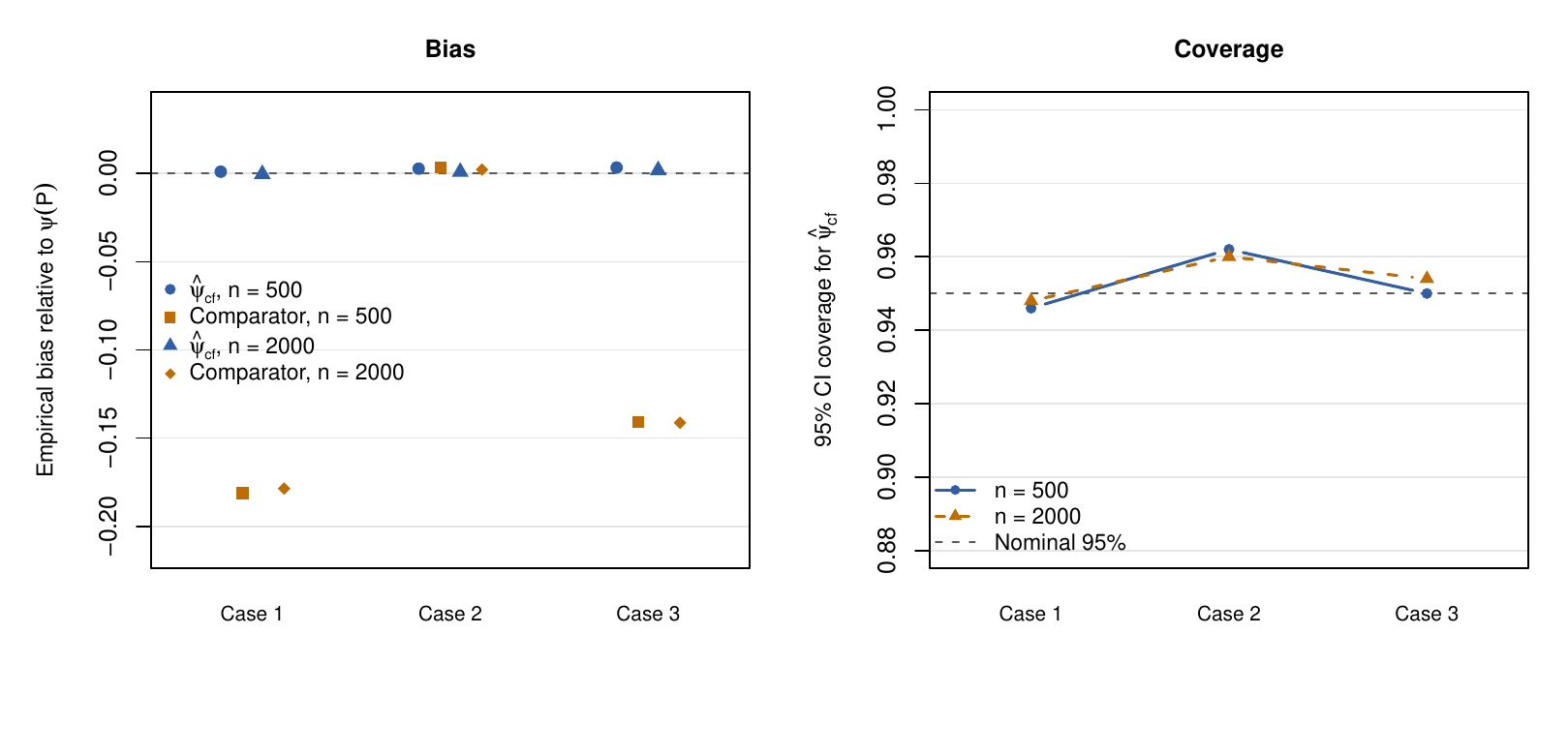}
\vspace{-10mm}
\caption{Empirical bias relative to $\psi(P)$ for the cross-fitted estimator and the mediation formula comparator, together with empirical 95\% confidence interval coverage for $\widehat\psi_{\mathrm{cf}}$.}
\label{fig:sim-summary}
\end{figure}

The simulation study highlights the distinction between numerical performance and target validity. A well-implemented mediation analysis can still answer a different question from the one of interest when ambiguity about the structural role of $W$ has not been resolved. The parameter $\psi(P)$ avoids that problem by keeping the target fixed across the confounder and mediator cases while also supporting standard doubly robust inference through the cross-fitted one-step estimator.

\section{Application: Direct Effect of Elevated PFAS Burden on Uric Acid Relative to Kidney Function}
\label{sec:application}

Recent work has linked per- and polyfluoroalkyl substances (PFAS) to higher uric acid levels and has suggested that kidney function decline may mediate part of this association \citep{niu2024kidney,zeng2019isomers}. This is a natural setting for our proposal because the structural role of kidney function is ambiguous in cross-sectional PFAS studies: on the one hand, impaired kidney function can raise measured serum PFAS concentrations through reduced renal clearance and is itself a strong determinant of serum uric acid, so estimated glomerular filtration rate (eGFR) may be viewed as a pre-exposure characteristic that confounds the PFAS--uric acid association \citep{kshirsagar2022environmental,dhingra2017reverse,moon2021pfas_kidney}. On the other hand, higher PFAS burden has also been associated with lower kidney function, so eGFR may plausibly lie on the pathway from PFAS burden to uric acid \citep{lin2021pfas_kidney,niu2024kidney}. Recent studies have further noted that PFAS--uric acid associations vary across kidney function stages \citep{zeng2019isomers,niu2024kidney}. Thus kidney function is scientifically relevant but structurally ambiguous, aligned with the setting considered in this paper.

We used publicly available NHANES 2013--2018 data from the demographic, laboratory, examination, questionnaire, and environmental-chemicals components \citep{nhaneswebsite}. We restricted attention to adults aged 20 years or older and used a complete-case analysis for the variables described below, leaving an analytic sample of $n=4,287$. Because our method is developed for a binary exposure, we defined
\[
A=\one\{\mathrm{PFHxS}+\mathrm{PFNA}+\mathrm{PFOA}+\mathrm{PFOS}\ge 20\text{ ng/mL}\}.
\]
The 20 ng/mL cutoff is the NASEM/ATSDR increased-risk category. When the available panel includes PFOA, PFOS, PFHxS, and PFNA, NASEM recommends forming the sum from those committee analytes \citep{nasem2022pfas,atsdr2024pfas}. We set $W$ equal to eGFR and $Y$ equal to serum uric acid. eGFR was computed from serum creatinine using the 2021 CKD-EPI creatinine equation \citep{inker2021new}, and the NHANES creatinine assay is standardized to an isotope-dilution mass-spectrometry reference method, which is the calibration assumed by that equation \citep{nhanes_biopro_2015}. The baseline covariates $X$ were age, sex, race/ethnicity, education, poverty-income ratio, body mass index, log serum cotinine, physical activity, and survey cycle. In the weighted analytic sample, 9.69\% met the high-PFAS definition, the mean PFAS sum was 10.90 ng/mL, the mean eGFR was 96.07 mL/min/1.73m$^2$, and the mean serum uric acid was 5.43 mg/dL.

Because serum PFAS were measured in a one-third subsample of NHANES participants aged 12 years and older, and because NHANES uses a multistage stratified design, we used the PFAS subsample weights together with the NHANES strata and primary sampling units (PSUs) \citep{nhanes_pfas_2015,nhanes_sample_design}. Following CDC guidance for combining post-2001 NHANES cycles, we formed a 6-year PFAS weight by dividing the 2-year PFAS subsample weight by 3 \citep{nhanes_weighting_tutorial}; see Appendix~\ref{app:nhanes-weights} for details. We estimated the survey-weighted analogue $\widehat\psi_{\mathrm{cf},w}$ of $\psi(P)$ with the 5-fold cross-fitted one-step estimator of Section~4.2. The outcome regression included $A$, $W$, an $A\times W$ interaction, and all components of $X$. The exposure model included $W$ and all components of $X$, so the nuisance specifications matched the covariate set used in the application. Predicted exposure probabilities were truncated to the interval $[0.01,0.99]$. As a mediation model comparator, we estimated the natural direct effect using weighted linear models for $W$ given $A$ and $X$ and for $Y$ given $A$, $W$, $A\times W$, and $X$. For the model-robust direct effect, we report both a design-based confidence interval based on the survey-weighted mean of the cross-fitted one-step scores and a percentile confidence interval from a PSU bootstrap with $500$ replicates. For both the model-robust estimator and the mediation model comparator, the PSU bootstrap resampled primary sampling units with replacement within strata and refit all nuisance models in every replicate. We do not report a design-based confidence interval for the comparator because a simple survey-linearization of the weighted mean of fitted comparator contrasts treats the nuisance regressions as fixed and therefore fails to propagate nuisance estimation uncertainty; the PSU bootstrap does propagate that uncertainty. On the other hand, for our cross-fitted one-step estimator, nuisance fitting contributes only a second-order remainder.

The estimated model-robust direct effect was
\(
\widehat\psi_{\mathrm{cf},w}=0.25\text{ mg/dL},
\)
with design-based 95\% confidence interval
\(
(0.05,\ 0.45)
\)
and percentile PSU-bootstrap 95\% confidence interval
\(
(0.12,\ 0.45).
\)
The mediation model comparator was
\(
\widehat\theta_{\mathrm{NDE},w}=0.10\text{ mg/dL},
\)
with percentile PSU-bootstrap 95\% confidence interval
\(
(-0.03,\ 0.23).
\)
Thus the model-robust direct effect exceeded the mediation model comparator by 0.15 mg/dL, or about 59.30\% of the estimated model-robust direct effect. 

\begin{table}[t]\small
\centering
\caption{Survey-weighted NHANES PFAS application results. The mediation model comparator corresponds to the natural direct effect under the mediator model.}
\label{tab:nhanes-app}
\begin{tabular}{lccc}
\toprule
Quantity & Estimate (mg/dL) & Design-based 95\% CI & PSU-bootstrap 95\% CI \\
\midrule
Model-robust direct effect $\psi(P)$ & 0.25 & $(0.05,\ 0.45)$ & $(0.12,\ 0.45)$ \\
Mediation model comparator & 0.10 & --- & $(-0.03,\ 0.23)$ \\
Difference $\widehat\psi_{\mathrm{cf},w}-\widehat\theta_{\mathrm{NDE},w}$ & 0.15 & --- & $(0.05,\ 0.30)$ \\
\bottomrule
\end{tabular}
\end{table}

This application illustrates the practical point of our framework. In this PFAS--kidney function--uric acid example, the mediation model comparator yields a materially smaller reported direct effect than the model-robust procedure. If kidney function is genuinely post-exposure, then the comparator has the usual natural direct effect interpretation; if kidney function is instead better viewed as pre-exposure because reduced renal clearance elevates serum PFAS, that interpretation no longer holds, whereas $\psi(P)$ still does. The application shows that when the structural role of kidney function is ambiguous, mediation-style analyses and the model-robust procedure proposed here can lead to materially different reported answers. In that sense, the model-robust direct effect provides a single interpretable summary that remains valid under either causal ordering.

\section{Conclusion}\label{sec:conclusion}

This paper studies a basic but overlooked problem in direct effect analysis: what should be estimated when a focal variable is known to affect the outcome but may be either a pre-exposure confounder or a post-exposure mediator? The main negative result is that one generally cannot use a single observed-data functional to recover both the confounder-model direct effect and the mediator-model natural direct effect. As a consequence, if a practitioner applies natural direct effect methodology when the intermediate variable is actually pre-exposure, the resulting estimate may target a quantity with no clear causal interpretation. The main positive result is that the proposed parameter keeps the statistical target fixed across the two causal scenarios: it equals the average treatment effect when $W$ is pre-exposure and an interventional direct effect when $W$ is post-exposure. A further contribution is a canonicity result: within a natural class of outcome-free stochastic direct effects, $\psi(P)$ is unique. In that sense, $\psi(P)$ is not merely a convenient repair to the estimand mismatch problem; it is the canonical observed-data target in this class.

This contribution is about model robustness, or identification robustness, rather than the more familiar estimator-level robustness for a fixed estimand. The uniqueness result sharpens that contribution by showing that, within a natural class, the proposed target is not just one plausible bridge between the two causal scenarios, but the only one that coincides with the confounder-model direct effect for every law and therefore retains a direct effect interpretation under both models. At the same time, the inferential procedure is robust at two levels. The estimand is model-robust across the two causal scenarios, and the cross-fitted one-step estimator is doubly robust with respect to nuisance estimation. The simulation study shows that the estimator performs well and clarifies that discrepancies between mediation analyses and adjustment analyses can reflect estimand mismatch rather than estimation failure. The survey-weighted NHANES application on elevated PFAS burden, kidney function, and uric acid illustrates this point in practice: a mediation model analysis treating kidney function as post-exposure produced a materially smaller reported direct-effect estimate than the proposed model-robust procedure, even though kidney function may plausibly be either pre-exposure or post-exposure in cross-sectional PFAS data.

Several extensions deserve future work. The present analysis is formulated on the additive mean scale with a single intermediate variable. Longitudinal settings, survival outcomes, and effect measures on nonadditive scales will require additional work. It would also be useful to study how sensitivity analysis for structural ambiguity can be combined with sensitivity analysis for unmeasured confounding and how the present ideas can be developed more fully for complex survey settings.

\bibliographystyle{plainnat}
\bibliography{revised_manuscript}

\newpage

\begin{center}

\textbf{\Large Supplementary Material for\\ ``Model-Robust Direct Effect Under\\ Confounder--Mediator Ambiguity''\vspace{5mm}}

AmirEmad Ghassami

Department of Mathematics and Statistics, Boston University

\vspace{10mm}

\end{center}

\appendix

\section{Proofs of Section \ref{sec:desc}}

\subsection{Proof of Proposition~\ref{prop:confounder-id}}

For each $a\in\{0,1\}$,
\[
\E\!\left\{Y^{(a)}\right\}
=
\E\!\left[\E\!\left\{Y^{(a)}\mid X,W\right\}\right].
\]
By Assumption~\ref{ass:confounder}(b),
\[
\E\!\left\{Y^{(a)}\mid X,W\right\}
=
\E\!\left\{Y^{(a)}\mid A=a,X,W\right\},
\]
and by Assumption~\ref{ass:confounder}(a),
\[
\E\!\left\{Y^{(a)}\mid A=a,X,W\right\}
=
\E\!\left\{Y\mid A=a,X,W\right\}
=
Q(a,W,X).
\]
Therefore
\[
\E\!\left\{Y^{(a)}\right\}=\E_P\!\left\{Q(a,W,X)\right\}.
\]
Subtracting the expression for $a=0$ from the expression for $a=1$ gives
\[
\theta_{\mathrm{ATE}}
=
\E_P\!\left\{Q(1,W,X)-Q(0,W,X)\right\}
=
\E_P\!\left\{\Delta(W,X)\right\}.
\]
The equivalent integral representation follows from iterated expectation:
\[
\E_P\!\left\{\Delta(W,X)\right\}
=
\E_P\!\left[\int \Delta(w,X)\,\dd F_{W\mid X}(w\mid X)\right].
\qedhere
\]

\subsection{Proof of Proposition~\ref{prop:nde-id}}

By iterated expectation,
\[
\theta_{\mathrm{NDE}}
=
\E\!\left[\E\!\left\{Y^{(1,W^{(0)})}-Y^{(0,W^{(0)})}\mid X\right\}\right].
\]
Fix $x$. Under Assumption~\ref{ass:crossworld},
\[
\E\!\left\{Y^{(a,W^{(0)})}\mid X=x\right\}
=
\int \E\!\left\{Y^{(a,w)}\mid X=x\right\}\,\dd F_{W^{(0)}\mid X}(w\mid x).
\]
By Assumption~\ref{ass:mediator}(c),
\[
\E\!\left\{Y^{(a,w)}\mid X\right\}
=
\E\!\left\{Y^{(a,w)}\mid A=a,X\right\}.
\]
Then Assumption~\ref{ass:mediator}(d) implies
\[
\E\!\left\{Y^{(a,w)}\mid A=a,X\right\}
=
\E\!\left\{Y^{(a,w)}\mid A=a,W=w,X\right\},
\]
and Assumption~\ref{ass:mediator}(b) yields
\[
\E\!\left\{Y^{(a,w)}\mid A=a,W=w,X\right\}
=
\E\!\left\{Y\mid A=a,W=w,X\right\}
=
Q(a,w,X).
\]
Finally, Assumption~\ref{ass:mediator}(c) together with Assumption~\ref{ass:mediator}(a) gives
\[
F_{W^{(0)}\mid X}(\cdot\mid X)=F_{W\mid A=0,X}(\cdot\mid 0,X).
\]
Combining the pieces,
\[
\theta_{\mathrm{NDE}}
=
\E_P\!\left[\int \{Q(1,w,X)-Q(0,w,X)\}\,\dd F_{W\mid A=0,X}(w\mid 0,X)\right],
\]
which is the stated formula.

\subsection{Proof of Theorem~\ref{thm:no-common}}

We construct 2 causal data-generating processes that induce the same observed-data law but different direct-effect parameters.

\paragraph{Confounder model.}
Let $U_W,U_A\stackrel{\mathrm{iid}}{\sim}\mathrm{Unif}(0,1)$. Define
\[
W=\one\{U_W\le 1/2\},
\qquad
A=\one\{U_A\le 1/4 + W/2\},
\qquad
Y=AW.
\]
Then $W$ is pre-exposure. The potential outcomes are
\[
Y^{(0)}=0,\qquad Y^{(1)}=W,
\]
so
\[
\theta_{\mathrm{ATE}}=\E(W)=1/2.
\]
Because $Y^{(a)}$ is a measurable function of $W$ alone, $Y^{(a)}\indep A\mid W$ holds. Thus Assumption~\ref{ass:confounder} is satisfied (with no baseline covariates $X$).

\paragraph{Mediator model.}
Let $U_A,U_W\stackrel{\mathrm{iid}}{\sim}\mathrm{Unif}(0,1)$. Define
\[
A=\one\{U_A\le 1/2\},
\qquad
W=\one\{U_W\le 1/4 + A/2\},
\qquad
Y=AW.
\]
Then $W$ is post-exposure. The potential mediator values are
\[
W^{(0)}=\one\{U_W\le 1/4\},
\qquad
W^{(1)}=\one\{U_W\le 3/4\},
\]
and the potential outcomes satisfy
\[
Y^{(0,w)}=0,\qquad
Y^{(1,w)}=w.
\]
Hence
\[
Y^{(1,W^{(0)})}=W^{(0)},\qquad
Y^{(0,W^{(0)})}=0,
\]
so
\[
\theta_{\mathrm{NDE}}=\E\!\left\{W^{(0)}\right\}=1/4.
\]
Because the structural equations use independent exogenous variables, Assumptions~\ref{ass:mediator} and \ref{ass:crossworld} hold.

\paragraph{Observed-data law.}
Under the confounder model,
\[
\Prob(W=1)=1/2,\qquad
\Prob(A=1\mid W=0)=1/4,\qquad
\Prob(A=1\mid W=1)=3/4.
\]
Under the mediator model,
\[
\Prob(A=1)=1/2,\qquad
\Prob(W=1\mid A=0)=1/4,\qquad
\Prob(W=1\mid A=1)=3/4.
\]
These imply the same joint distribution of $(A,W)$:
\begin{align*}
\Prob(A=0,W=0)&=3/8,\\
\Prob(A=1,W=0)&=1/8,\\
\Prob(A=0,W=1)&=1/8,\\
\Prob(A=1,W=1)&=3/8.
\end{align*}
Since in both models $Y=AW$ deterministically, the full observed-data law of $(A,W,Y)$ is the same.

Thus a common observed-data functional $T(P)$ would have to equal both $1/2$ and $1/4$ for the same observed-data law $P$, which is impossible.

\section{Proofs of Section \ref{sec:no-interaction}}
\subsection{Proof of Theorem~\ref{thm:no-interaction}}

Under Assumption~\ref{ass:no-interaction},
\[
\Delta(w,x)=\delta(x)
\]
for $P$-almost every $(w,x)$. Therefore
\[
\E_P\!\left\{\Delta(W,X)\right\}
=
\E_P\!\left\{\delta(X)\right\}.
\]
Likewise,
\[
\E_P\!\left[\int \Delta(w,X)\,\dd F_{W\mid X}(w\mid X)\right]
=
\E_P\!\left\{\delta(X)\right\},
\]
and
\[
\E_P\!\left[\int \Delta(w,X)\,\dd F_{W\mid A=0,X}(w\mid 0,X)\right]
=
\E_P\!\left\{\delta(X)\right\}.
\]
Now apply Proposition~\ref{prop:confounder-id} to the confounder model direct effect and Proposition~\ref{prop:nde-id} to the mediator-model natural direct effect.

\subsection{Proof of Proposition~\ref{prop:sensitivity}}

By Proposition~\ref{prop:nde-id} and the definition of $\psi(P)$,
\[
\psi(P)-\theta_{\mathrm{NDE}}
=
\E_P\!\left[
\int \Delta(w,X)\,
\Bigl\{
\dd F_{W\mid X}(w\mid X)-\dd F_{W\mid A=0,X}(w\mid 0,X)
\Bigr\}
\right].
\]

Now let $\pi(X):=\Prob(A=1\mid X)$. By the law of total probability,
\[
F_{W\mid X}(\cdot\mid X)
=
\pi(X)\,F_{W\mid A=1,X}(\cdot\mid 1,X)
+
\{1-\pi(X)\}\,F_{W\mid A=0,X}(\cdot\mid 0,X).
\]
Substituting this identity into the previous display gives
\[
\psi(P)-\theta_{\mathrm{NDE}}
=
\E_P\!\left[
\pi(X)\int \Delta(w,X)\,
\Bigl\{
\dd F_{W\mid A=1,X}(w\mid 1,X)-\dd F_{W\mid A=0,X}(w\mid 0,X)
\Bigr\}
\right].
\]
Recognizing the two conditional integrals as
\[
\E_P\{\Delta(W,X)\mid A=a,X\},
\qquad a\in\{0,1\},
\]
yields \eqref{eq:gap-conditional}.

If $W$ is binary, then
\begin{align*}
&\int \Delta(w,X)\,
\Bigl\{
\dd F_{W\mid X}(w\mid X)-\dd F_{W\mid A=0,X}(w\mid 0,X)
\Bigr\} \\
&\qquad=
\Delta(1,X)\Bigl\{\Prob(W=1\mid X)-\Prob(W=1\mid A=0,X)\Bigr\} \\
&\qquad\quad+
\Delta(0,X)\Bigl\{\Prob(W=0\mid X)-\Prob(W=0\mid A=0,X)\Bigr\} \\
&\qquad=
\bigl\{\Delta(1,X)-\Delta(0,X)\bigr\}
\Bigl\{\Prob(W=1\mid X)-\Prob(W=1\mid A=0,X)\Bigr\},
\end{align*}
which proves the binary simplification.

For the bound, fix $x$ and let $P_x$ and $Q_x$ denote the conditional laws
\[
P_x=F_{W\mid X}(\cdot\mid x),
\qquad
Q_x=F_{W\mid A=0,X}(\cdot\mid 0,x).
\]
For the bounded measurable function $h_x(w)=\Delta(w,x)$,
\[
\int h_x(w)\,\dd P_x(w)-\int h_x(w)\,\dd Q_x(w)
\]
is the conditional discrepancy at $x$. A standard total-variation inequality yields
\[
\left|
\int h_x(w)\,\dd P_x(w)-\int h_x(w)\,\dd Q_x(w)
\right|
\le
\operatorname{rng}_{\Delta}(x)\,\TV(x).
\]
Integrating over the marginal distribution of $X$ gives
\[
\left|\theta_{\mathrm{NDE}}-\psi(P)\right|
\le
\E_P\!\left\{\operatorname{rng}_{\Delta}(X)\,\TV(X)\right\}.
\]
If $\operatorname{rng}_{\Delta}(X)\le \Gamma$ almost surely, then
\[
\left|\theta_{\mathrm{NDE}}-\psi(P)\right|
\le
\Gamma\,\E_P\{\TV(X)\},
\]
which is equivalent to the stated interval.

\section{Proofs of Section \ref{sec:psi}}

\subsection{Proof of Theorem~\ref{thm:psi}}

Under the confounder model, Proposition~\ref{prop:confounder-id} already shows that
\[
\psi(P)=\E_P\!\left\{\Delta(W,X)\right\}=\theta_{\mathrm{ATE}}.
\]

Under the mediator model, conditional on $X$, the stochastic intervention variable $G$ has law $F_{W\mid X}(\cdot\mid X)$ and is independent of the collection $\{Y^{(a,w)}:a\in\{0,1\},w\in\mathcal{W}\}$ given $X$. Therefore, for each $a\in\{0,1\}$,
\[
\E\!\left\{Y^{(a,G)}\mid X\right\}
=
\int \E\!\left\{Y^{(a,w)}\mid X\right\}\,\dd F_{W\mid X}(w\mid X).
\]
Under Assumption~\ref{ass:mediator}, the same argument used in the proof of Proposition~\ref{prop:nde-id} gives
\[
\E\!\left\{Y^{(a,w)}\mid X\right\}=Q(a,w,X).
\]
Hence
\[
\E\!\left\{Y^{(a,G)}\mid X\right\}
=
\int Q(a,w,X)\,\dd F_{W\mid X}(w\mid X).
\]
Subtracting the $a=0$ expression from the $a=1$ expression and taking expectations over $X$ yields
\[
\theta_{\mathrm{IDE}}
=
\E_P\!\left[\int \{Q(1,w,X)-Q(0,w,X)\}\,\dd F_{W\mid X}(w\mid X)\right]
=
\E_P\!\left\{\Delta(W,X)\right\}
=
\psi(P).
\qedhere
\]

\subsection{Proof of Theorem~\ref{thm:psi-unique}}

Fix an observed-data law $P$ satisfying Assumption~\ref{ass:positivity}, and let $P_{XAW}$ denote its marginal law on $(X,A,W)$. We will show that
\[
\nu_P(\cdot\mid X)=F_{W\mid X}(\cdot\mid X)
\qquad P\text{-almost surely}.
\]

Let $h$ be any bounded measurable function of $(W,X)$. Construct an auxiliary observed-data law $P^h$ by taking $(X,A,W)\sim P_{XAW}$ and setting
\[
Y:=A\,h(W,X).
\]
Because $P^h$ has the same $(X,A,W)$-margin as $P$, Assumption~\ref{ass:positivity} holds under $P^h$ as well. Moreover, $P^h$ is compatible with the confounder model: take $W$ to be pre-exposure and define
\[
Y^{(0)}=0,
\qquad
Y^{(1)}=h(W,X).
\]
Then $Y^{(a)}$ is measurable with respect to $(X,W)$, so $Y^{(a)}\indep A\mid (X,W)$ for $a\in\{0,1\}$, and Assumption~\ref{ass:confounder} holds. Under $P^h$,
\[
Q(1,w,x)=h(w,x),
\qquad
Q(0,w,x)=0,
\]
hence
\[
\Delta(w,x)=h(w,x).
\]

Since the rule $P\mapsto \nu_P$ is outcome-free and $P^h_{XAW}=P_{XAW}$, we have
\[
\nu_{P^h}(\cdot\mid x)=\nu_P(\cdot\mid x)
\]
for $P_X$-almost every $x$. By the assumed confounder model property of $\Psi_\nu$,
\[
\E_{P^h}\!\left[\int h(w,X)\,\dd \nu_P(w\mid X)\right]
=
\Psi_\nu(P^h)
=
\theta_{\mathrm{ATE}}
=
\E_{P^h}\!\left\{Y^{(1)}-Y^{(0)}\right\}
=
\E_{P^h}\!\left\{h(W,X)\right\}.
\]
Because $P^h$ and $P$ have the same marginal law of $(X,W)$, this is equivalent to
\[
\E_P\!\left[\int h(w,X)\,\dd \nu_P(w\mid X)\right]
=
\E_P\!\left\{h(W,X)\right\}
\]
for every bounded measurable $h$.

By the defining property of the conditional law $F_{W\mid X}$,
\[
\E_P\!\left\{h(W,X)\right\}
=
\E_P\!\left[\int h(w,X)\,\dd F_{W\mid X}(w\mid X)\right].
\]
Therefore,
\[
\E_P\!\left[\int h(w,X)\,\dd \nu_P(w\mid X)\right]
=
\E_P\!\left[\int h(w,X)\,\dd F_{W\mid X}(w\mid X)\right]
\]
for every bounded measurable $h$.

Now fix any measurable set $B \subseteq \mathcal{W}$ and any event $E \in \sigma(X)$. Taking
\[
h(w,x)=\mathbf{1}\{w\in B\}\mathbf{1}\{x\in E\},
\]
we obtain
\[
\E_P\!\left[\mathbf{1}\{X\in E\}\,\nu_P(B\mid X)\right]
=
\E_P\!\left[\mathbf{1}\{X\in E\}\,F_{W\mid X}(B\mid X)\right].
\]
Equivalently, with
\[
D_B(X):=\nu_P(B\mid X)-F_{W\mid X}(B\mid X),
\]
we have
\[
\E_P\!\left[\mathbf{1}\{X\in E\}\,D_B(X)\right]=0
\qquad\text{for every }E\in \sigma(X).
\]
Since $D_B(X)$ is $\sigma(X)$-measurable, choosing
\[
E=\{D_B(X)>0\}
\qquad\text{and}\qquad
E=\{D_B(X)<0\}
\]
shows that $D_B(X)=0$ almost surely. Hence, for every measurable $B\subseteq\mathcal{W}$,
\[
\nu_P(B\mid X)=F_{W\mid X}(B\mid X)
\qquad P\text{-almost surely}.
\]
Thus $\nu_P(\cdot\mid X)$ is a version of the conditional law of $W$ given $X$; by uniqueness of regular conditional distributions up to $P_X$-null sets,
\[
\nu_P(\cdot\mid X)=F_{W\mid X}(\cdot\mid X)
\qquad P\text{-almost surely}.
\]

Substituting this identity into the definition of $\Psi_\nu(P)$ yields
\[
\Psi_\nu(P)
=
\E_P\!\left[\int \Delta(w,X)\,\dd F_{W\mid X}(w\mid X)\right]
=
\E_P\!\left\{\Delta(W,X)\right\}
=
\psi(P).
\]
Because $P$ was arbitrary, the conclusion holds for every observed-data law satisfying Assumption~\ref{ass:positivity}.

\subsection{Proof of Theorem~\ref{thm:eif}}

Let $Z=(W,X)$. Define
\[
Q_a(z):=\E_P(Y\mid A=a,Z=z), \qquad g(z):=\Prob(A=1\mid Z=z),
\]
and note that
\[
\psi(P)=\E_P\!\left\{Q_1(Z)-Q_0(Z)\right\}.
\]

Consider a regular one-dimensional parametric submodel $\{P_{\varepsilon}:\varepsilon\in(-\delta,\delta)\}$ through $P=P_0$ with score $S(O)$. Factor the submodel density as
\[
p_{\varepsilon}(o)
=
p_{\varepsilon}(y\mid a,z)\,
p_{\varepsilon}(a\mid z)\,
p_{\varepsilon}(z).
\]
Write
\[
S(O)=S_Y(O)+S_A(O)+S_Z(O),
\]
where each term is the score for the corresponding factor.

Let $m_{\varepsilon}(z)=Q_{1,\varepsilon}(z)-Q_{0,\varepsilon}(z)$. Then
\[
\psi(P_{\varepsilon})=\E_{\varepsilon}\!\left\{m_{\varepsilon}(Z)\right\},
\]
so the derivative at $\varepsilon=0$ is
\[
\left.\frac{d}{d\varepsilon}\psi(P_{\varepsilon})\right|_{\varepsilon=0}
=
\E\!\left[m(Z)S_Z(O)\right]
+
\E\!\left\{\dot m(Z)\right\},
\]
where a dot denotes differentiation at $\varepsilon=0$ and $m(z)=Q_1(z)-Q_0(z)$.

For $a=1$,
\[
\dot Q_1(z)=\E\!\left[\{Y-Q_1(z)\}S_Y(O)\mid A=1,Z=z\right],
\]
and therefore
\[
\E\!\left\{\dot Q_1(Z)\right\}
=
\E\!\left[\frac{\one\{A=1\}}{g(Z)}\{Y-Q_1(Z)\}S_Y(O)\right].
\]
Similarly,
\[
\E\!\left\{\dot Q_0(Z)\right\}
=
\E\!\left[\frac{\one\{A=0\}}{1-g(Z)}\{Y-Q_0(Z)\}S_Y(O)\right].
\]
Thus
\begin{align*}
\left.\frac{d}{d\varepsilon}\psi(P_{\varepsilon})\right|_{\varepsilon=0}
=
\E\Bigg[
&\left\{
\frac{\one\{A=1\}}{g(Z)}(Y-Q_1(Z))
-
\frac{\one\{A=0\}}{1-g(Z)}(Y-Q_0(Z))
\right\}S_Y(O)
\\
&\qquad\qquad
+\{m(Z)-\psi(P)\}S_Z(O)
\Bigg].
\end{align*}
Because $\psi(P)$ does not depend on the conditional law of $A$ given $Z$, the derivative in the $S_A$ direction is $0$. Also,
\[
\E\!\left[\varphi_{\psi}(O)\mid Z\right]=m(Z)-\psi(P),
\]
so $\E\{\varphi_{\psi}(O)S_A(O)\}=0$. Hence
\[
\left.\frac{d}{d\varepsilon}\psi(P_{\varepsilon})\right|_{\varepsilon=0}
=
\E\!\left\{\varphi_{\psi}(O)S(O)\right\},
\]
with $\varphi_{\psi}(O)$ as given in \eqref{eq:eif}. Since the nonparametric tangent space is the full mean-zero $L_2(P)$ space, this influence function is the canonical gradient, i.e., the efficient influence function.

\subsection{A Remainder Formula for the One-Step Score}

Write $Q_{a,0}(w,x):=Q(a,w,x)$ for $a\in\{0,1\}$ and
\(
g_0(w,x):=\Prob(A=1\mid W=w,X=x)
\)
for the true nuisance functions, and let
\(
\eta_0:=(Q_{1,0},Q_{0,0},g_0).
\)
For any candidate nuisance collection $\eta=(Q_1,Q_0,g)$ such that $0<g(W,X)<1$ almost surely, define
\begin{align*}
\varphi_{\eta}(O)
:=
&\frac{\one\{A=1\}}{g(W,X)}\left\{Y-Q_1(W,X)\right\}
-
\frac{\one\{A=0\}}{1-g(W,X)}\left\{Y-Q_0(W,X)\right\}
\\
&\qquad
+Q_1(W,X)-Q_0(W,X)-\psi(P).
\end{align*}
Note that $\varphi_{\eta_0}$ is exactly the efficient influence function in \eqref{eq:eif}. Also, throughout the proof, we use the shorthand
\(
Pf := \E_P\{f(O)\}.
\)

\begin{lemma}[Second-order remainder]\label{lem:remainder}
For any nuisance functions $\eta=(Q_1,Q_0,g)$,
\begin{align*}
P\varphi_{\eta}
=
&
P\!\left[
\frac{g_0(W,X)-g(W,X)}{g(W,X)}
\{Q_{1,0}(W,X)-Q_1(W,X)\}
\right]
\\
&+
P\!\left[
\frac{g_0(W,X)-g(W,X)}{1-g(W,X)}
\{Q_{0,0}(W,X)-Q_0(W,X)\}
\right].
\end{align*}
\end{lemma}

\begin{proof}
Since
\[
\psi(P)=P\!\left\{Q_{1,0}(W,X)-Q_{0,0}(W,X)\right\},
\]
it suffices to compare the candidate nuisance functions with the truth.

For the $a=1$ term,
\[
P\!\left[
\frac{\one\{A=1\}}{g(W,X)}\{Y-Q_1(W,X)\}
\,\middle|\, W,X
\right]
=
\frac{g_0(W,X)}{g(W,X)}
\{Q_{1,0}(W,X)-Q_1(W,X)\}.
\]
Therefore,
\begin{align*}
&P\!\left[
\frac{\one\{A=1\}}{g(W,X)}\{Y-Q_1(W,X)\}
+
Q_1(W,X)-Q_{1,0}(W,X)
\right]
\\
&\qquad=
P\!\left[
\left\{\frac{g_0(W,X)}{g(W,X)}-1\right\}
\{Q_{1,0}(W,X)-Q_1(W,X)\}
\right]
\\
&\qquad=
P\!\left[
\frac{g_0(W,X)-g(W,X)}{g(W,X)}
\{Q_{1,0}(W,X)-Q_1(W,X)\}
\right].
\end{align*}

For the $a=0$ term,
\[
P\!\left[
\frac{\one\{A=0\}}{1-g(W,X)}\{Y-Q_0(W,X)\}
\,\middle|\, W,X
\right]
=
\frac{1-g_0(W,X)}{1-g(W,X)}
\{Q_{0,0}(W,X)-Q_0(W,X)\}.
\]
Hence,
\begin{align*}
&P\!\left[
-\frac{\one\{A=0\}}{1-g(W,X)}\{Y-Q_0(W,X)\}
-Q_0(W,X)+Q_{0,0}(W,X)
\right]
\\
&\qquad=
-
P\!\left[
\frac{1-g_0(W,X)}{1-g(W,X)}
\{Q_{0,0}(W,X)-Q_0(W,X)\}
\right]
+
P\!\left\{Q_{0,0}(W,X)-Q_0(W,X)\right\}
\\
&\qquad=
P\!\left[
\left\{1-\frac{1-g_0(W,X)}{1-g(W,X)}\right\}
\{Q_{0,0}(W,X)-Q_0(W,X)\}
\right]
\\
&\qquad=
P\!\left[
\frac{g_0(W,X)-g(W,X)}{1-g(W,X)}
\{Q_{0,0}(W,X)-Q_0(W,X)\}
\right].
\end{align*}

Adding the two identities proves the claim.

\end{proof}

\subsection{Proof of Theorem~\ref{thm:dr}}

For each fold $k\in\{1,\ldots,K\}$, let $n_k:=|\mathcal I_k|$ and define the fold-specific empirical measure
\[
\mathbb P_{n,k}f
:=
\frac{1}{n_k}\sum_{i\in\mathcal I_k} f(O_i).
\]
Also write
\[
\widehat\eta^{(-k)}
:=
\bigl(\widehat Q^{(-k)}_1,\widehat Q^{(-k)}_0,\widehat g^{(-k)}\bigr),
\]
where $\widehat Q^{(-k)}_a(w,x):=\widehat Q^{(-k)}(a,w,x)$ for $a\in\{0,1\}$.
Then, by definition of the cross-fitted estimator,
\[
\widehat\psi_{\mathrm{cf}}-\psi(P)
=
\sum_{k=1}^K \frac{n_k}{n}\,\mathbb P_{n,k}\varphi_{\widehat\eta^{(-k)}}.
\]

Let
\[
\mathcal E_n
:=
\left\{
\varepsilon \le \widehat g^{(-k)}(w,x)\le 1-\varepsilon
\text{ for all }(w,x)\text{ and all }k=1,\ldots,K
\right\}.
\]
By assumption, $\Prob(\mathcal E_n)\to 1$.

\medskip
\noindent
\textit{Proof of part (i).}
Write
\[
\widehat\psi_{\mathrm{cf}}-\psi(P)
=
\sum_{k=1}^K R_{n,k}
+
\sum_{k=1}^K B_{n,k},
\]
where
\[
R_{n,k}
:=
\frac{n_k}{n}\,(\mathbb P_{n,k}-P)\varphi_{\widehat\eta^{(-k)}},
\qquad
B_{n,k}
:=
\frac{n_k}{n}\,P\varphi_{\widehat\eta^{(-k)}}.
\]

Fix $k$. Conditional on the training sample
\[
\mathcal D_{-k}:=\{O_j:j\in\mathcal I_{-k}\},
\]
the nuisance estimators $\widehat\eta^{(-k)}$ are fixed, and the validation observations
$\{O_i:i\in\mathcal I_k\}$ are independent of $\mathcal D_{-k}$. Therefore,
\[
\E\!\left(R_{n,k}\mid \mathcal D_{-k}\right)=0
\]
and
\[
\E\!\left(R_{n,k}^2\mid \mathcal D_{-k}\right)
=
\frac{n_k}{n^2}
\Var_P\!\left\{\varphi_{\widehat\eta^{(-k)}}(O)\mid \mathcal D_{-k}\right\}
\le
\frac{n_k}{n^2}
P\!\left\{\varphi_{\widehat\eta^{(-k)}}(O)^2\right\}.
\]
On the event $\mathcal E_n$, the denominators in $\varphi_{\widehat\eta^{(-k)}}$ are uniformly bounded away from $0$ and $1$. Under the standing square-integrability conditions and the assumed $L_2(P)$ boundedness of the nuisance estimators, this implies
\[
P\!\left\{\varphi_{\widehat\eta^{(-k)}}(O)^2\right\}=O_P(1).
\]
Since $K$ is fixed and $n_k\asymp n$, it follows that
\[
R_{n,k}=o_P(1)
\qquad\text{for each }k,
\]
and hence
\[
\sum_{k=1}^K R_{n,k}=o_P(1).
\]

Next, on $\mathcal E_n$, Lemma~\ref{lem:remainder} and the Cauchy--Schwarz inequality give
\begin{align*}
\left|P\varphi_{\widehat\eta^{(-k)}}\right|
\le\;&
\left\|
\frac{g_0-\widehat g^{(-k)}}{\widehat g^{(-k)}}
\right\|_{L_2(P)}
\left\|Q_{1,0}-\widehat Q^{(-k)}_1\right\|_{L_2(P)}
\\
&+
\left\|
\frac{g_0-\widehat g^{(-k)}}{1-\widehat g^{(-k)}}
\right\|_{L_2(P)}
\left\|Q_{0,0}-\widehat Q^{(-k)}_0\right\|_{L_2(P)}
\\
\le\;&
\varepsilon^{-1}
\left\|\widehat g^{(-k)}-g_0\right\|_{L_2(P)}
\left(
\left\|Q_{1,0}-\widehat Q^{(-k)}_1\right\|_{L_2(P)}
+
\left\|Q_{0,0}-\widehat Q^{(-k)}_0\right\|_{L_2(P)}
\right).
\end{align*}
If the outcome regressions are consistent, then the last display is $o_P(1)$ because
\[
\left\|\widehat g^{(-k)}-g_0\right\|_{L_2(P)} \le 1.
\]
If instead $\widehat g^{(-k)}$ is consistent and the outcome regressions are uniformly bounded in $L_2(P)$, then the same display is again $o_P(1)$. Therefore,
\[
\sum_{k=1}^K B_{n,k}=o_P(1),
\]
and thus
\[
\widehat\psi_{\mathrm{cf}} \overset{P}{\longrightarrow} \psi(P).
\]

\medskip
\noindent
\textit{Proof of part (ii).}
Now write
\[
\widehat\psi_{\mathrm{cf}}-\psi(P)
=
\sum_{k=1}^K \frac{n_k}{n}(\mathbb P_{n,k}-P)\varphi_{\eta_0}
+
\sum_{k=1}^K S_{n,k}
+
\sum_{k=1}^K B_{n,k},
\]
where
\[
S_{n,k}
:=
\frac{n_k}{n}(\mathbb P_{n,k}-P)
\left\{
\varphi_{\widehat\eta^{(-k)}}-\varphi_{\eta_0}
\right\}.
\]
Because the folds partition the sample,
\[
\sum_{k=1}^K \frac{n_k}{n}(\mathbb P_{n,k}-P)\varphi_{\eta_0}
=
(\mathbb P_n-P)\varphi_{\eta_0}.
\]

We first show that the stochastic remainder $\sum_{k=1}^K S_{n,k}$ is negligible. Fix $k$. Conditional on $\mathcal D_{-k}$,
\[
\E\!\left(S_{n,k}\mid \mathcal D_{-k}\right)=0
\]
and
\[
\E\!\left(nS_{n,k}^2\mid \mathcal D_{-k}\right)
=
\frac{n_k}{n}
\Var_P\!\left(
\varphi_{\widehat\eta^{(-k)}}(O)-\varphi_{\eta_0}(O)
\,\middle|\, \mathcal D_{-k}
\right)
\le
\frac{n_k}{n}
P\!\left[
\left\{\varphi_{\widehat\eta^{(-k)}}(O)-\varphi_{\eta_0}(O)\right\}^2
\right].
\]
On $\mathcal E_n$, the map $\eta\mapsto \varphi_\eta$ is continuous in $L_2(P)$ under the standing square-integrability conditions. Hence the assumed $L_2(P)$ consistency of $\widehat Q^{(-k)}_a$ and $\widehat g^{(-k)}$ implies
\[
P\!\left[
\left\{\varphi_{\widehat\eta^{(-k)}}(O)-\varphi_{\eta_0}(O)\right\}^2
\right]
=
o_P(1).
\]
Therefore,
\[
\sqrt{n}\,S_{n,k}=o_P(1)
\qquad\text{for each }k,
\]
and, since $K$ is fixed,
\[
\sum_{k=1}^K S_{n,k}=o_P(n^{-1/2}).
\]

Next, by the bound above and the product-rate assumption,
\[
\max_{1\le k\le K}
\left|P\varphi_{\widehat\eta^{(-k)}}\right|
=
o_P(n^{-1/2}),
\]
so
\[
\sum_{k=1}^K B_{n,k}=o_P(n^{-1/2}).
\]

Combining the displays,
\[
\widehat\psi_{\mathrm{cf}}-\psi(P)
=
(\mathbb P_n-P)\varphi_{\eta_0}
+
o_P(n^{-1/2}).
\]
Since $\varphi_{\eta_0}=\varphi_\psi$ and $\Var_P\{\varphi_\psi(O)\}<\infty$, the central limit theorem yields
\[
\sqrt{n}\left\{\widehat\psi_{\mathrm{cf}}-\psi(P)\right\}
\overset{d}{\longrightarrow}
N\!\left(0,\Var_P\{\varphi_\psi(O)\}\right).
\]

Finally, for $i\in\mathcal I_k$, the quantity
\[
\widehat\phi^{\mathrm{cf}}_i-\widehat\psi_{\mathrm{cf}}
\]
is the fold-specific plug-in estimate of the centered influence function. By the same $L_2(P)$-consistency argument used above,
\[
\frac{1}{n}\sum_{k=1}^K\sum_{i\in\mathcal I_k}
\left[
\left(\widehat\phi^{\mathrm{cf}}_i-\widehat\psi_{\mathrm{cf}}\right)
-
\varphi_\psi(O_i)
\right]^2
=
o_P(1).
\]
Therefore, the sample variance of the estimated influence values is consistent for
\[
P\!\left\{\varphi_\psi(O)^2\right\}
=
\Var_P\{\varphi_\psi(O)\},
\]
and the stated standard-error estimator is asymptotically valid.

\subsection{Proof of Proposition~\ref{prop:union-test}}

Under $H_0^{\cup}$, either $(\mathcal{C},\theta_{\mathrm{ATE}}=0)$ holds or $(\mathcal{M},\theta_{\mathrm{NDE}}=0)$ holds.

If $(\mathcal{C},\theta_{\mathrm{ATE}}=0)$ holds, then validity of $p_C$ gives
\[
\Prob(p_C<\alpha)\le \alpha.
\]
Hence
\[
\Prob(p_{\max}<\alpha)
=
\Prob(p_C<\alpha,\ p_M<\alpha)
\le
\Prob(p_C<\alpha)
\le
\alpha.
\]

If $(\mathcal{M},\theta_{\mathrm{NDE}}=0)$ holds, the same argument using $p_M$ shows
\[
\Prob(p_{\max}<\alpha)\le \alpha.
\]

Therefore the test that rejects when $p_{\max}<\alpha$ has size at most $\alpha$ under the union null.

\section{Survey-Weighted Implementation and PSU Bootstrap for the NHANES Application}
\label{app:nhanes-weights}

The theory in Sections~\ref{sec:desc}--\ref{sec:psi} is developed for independent and identically distributed observations from an observed-data law $P$. NHANES, however, is a stratified multistage probability sample, and serum PFAS are measured only in a one-third subsample of participants aged 12 years or older \citep{nhanes_pfas_2015,nhanes_sample_design}. The application in Section~\ref{sec:application} therefore uses a survey-weighted analogue of the proposed estimator.

Let $\omega_i^{(2)}$ denote the PFAS 2-year subsample weight supplied by NHANES for participant $i$. Because we combined the 2013--2014, 2015--2016, and 2017--2018 cycles, we formed the 6-year PFAS weight
\[
\omega_i^{(6)}:=\frac{1}{3}\omega_i^{(2)},
\]
following the CDC rule for combining post-2001 NHANES cycles \citep{nhanes_weighting_tutorial}. For any measurable function $f$, define the H\'ajek-weighted empirical mean
\[
P_n^{w}f
:=
\frac{\sum_{i=1}^n \omega_i^{(6)} f(O_i)}
{\sum_{i=1}^n \omega_i^{(6)}}.
\]
If the design weights are correctly specified, then $P_n^{w}$ targets the survey-weighted target law rather than the unweighted empirical law. In this sense, the application can be viewed as estimating
\[
\psi(P^\star)
=
E_{P^\star}\!\left\{Q^\star(1,W,X)-Q^\star(0,W,X)\right\},
\]
where $P^\star$ denotes the population law represented by the weighted NHANES sample.

The survey-weighted analogue of the proposed one-step estimator is
\[
\widehat\psi_{\mathrm{cf},w}
=
P_n^{w}\!\left[
\frac{A}{\widehat g}(Y-\widehat Q_1)
-
\frac{1-A}{1-\widehat g}(Y-\widehat Q_0)
+
\widehat Q_1-\widehat Q_0
\right],
\]
where, for an observation in validation fold $k$, $\widehat Q_a := \widehat Q^{(-k)}(a,W,X)$ and $\widehat g := \widehat g^{(-k)}(W,X)$ denote the held-out predictions obtained from the corresponding training folds, as in Section~\ref{sec:IF}. Within each training fold, the proposed estimator $\widehat\psi_{\mathrm{cf},w}$ fits only the outcome regression and the exposure regression, using normalized survey weights to stabilize computation while preserving the weighted estimating equations. The mediator regression is not part of $\widehat\psi_{\mathrm{cf},w}$; it is fit separately only for the mediation model comparator $\widehat\theta_{\mathrm{NDE},w}$.

For $\widehat\psi_{\mathrm{cf},w}$, we report a design-based 95\% confidence interval that uses the NHANES strata and primary sampling units (PSUs). After computing the held-out one-step score for each participant, we treat those scores as the analysis variable and take their survey-weighted mean. We then estimated the standard error of the survey-weighted mean of the held-out one-step scores with the standard NHANES Taylor series linearization variance estimator, using the released masked variance strata (\texttt{SDMVSTRA}) and masked variance PSUs (\texttt{SDMVPSU}) together with the PFAS subsample weights. This choice is motivated by the i.i.d. asymptotic theory of Theorem~\ref{thm:dr}: in the i.i.d. setting, the proposed cross-fitted one-step estimator is asymptotically equivalent to an average of influence-function score contributions, so its leading variance is driven by the variability of those scores, while nuisance estimation error contributes only a higher-order remainder.

For the mediation-model comparator, we do not report a corresponding design-based 95\% confidence interval. The comparator is computed by first fitting the outcome and mediator regressions and then averaging the resulting fitted contrasts. A simple design-based variance calculation for that weighted mean would treat the fitted regressions as if they were fixed, and would therefore ignore the extra uncertainty created by estimating those regressions for nuisance functions. For this reason, we use a within-stratum PSU bootstrap for the comparator. In each bootstrap replicate, we resample PSUs with replacement within strata, refit the outcome and mediator regressions, and recompute the comparator. We use the same bootstrap strategy for $\widehat\psi_{\mathrm{cf},w}$, except that each replicate refits the outcome and exposure regressions before recomputing the one-step estimator.

Specifically, let $h$ index strata and let $j=1,\dots,m_h$ index the observed PSUs within stratum $h$. In bootstrap replicate $b$, we sample $m_h$ PSUs with replacement from the $m_h$ observed PSUs in each stratum. Let $C_{hj}^{(b)}$ denote the number of times PSU $j$ is selected in replicate $b$. The bootstrap survey weight for unit $i$ is then
\[
\omega_i^{(6,b)}
=
C_{h(i)j(i)}^{(b)}\,\omega_i^{(6)},
\]
where $h(i)$ and $j(i)$ denote the stratum and PSU containing unit $i$. Using these replicate weights, we refit the outcome and exposure regressions and recompute the survey-weighted one-step estimator, denoted $\widehat\psi_{\mathrm{cf},w}^{(b)}$, and we refit the outcome and mediator regressions and recompute the mediation model comparator, denoted $\widehat\theta_{\mathrm{NDE},w}^{(b)}$. Percentile 95\% confidence intervals were obtained from the empirical 2.5th and 97.5th percentiles of the bootstrap replicates. The same bootstrap draws were also used to form a confidence interval for $\widehat\psi_{\mathrm{cf},w}-\widehat\theta_{\mathrm{NDE},w}$.

\end{document}